\documentclass[11pt,authoryear]{elsarticle}
\usepackage[margin=2.3cm, includehead]{geometry}
\makeatletter
\def\ps@pprintTitle{%
  \let\@oddhead\@empty
  \let\@evenhead\@empty
  \let\@oddfoot\@empty
  \let\@evenfoot\@oddfoot
}
\makeatother
\usepackage{currfile}
\usepackage{mathrsfs}
\usepackage[mathscr]{eucal}
\DeclareMathAlphabet{\mathpzc}{OT1}{pzc}{m}{it}
\usepackage{natbib}
\usepackage{makeidx}
\usepackage{multirow}
\usepackage{multicol}
\usepackage{booktabs}
\usepackage{enumerate}
\usepackage{setspace}
\usepackage{float}
\usepackage{hhline}
\usepackage{bbm}
\usepackage[dvipsnames,svgnames,table]{xcolor}
\usepackage{graphicx}
\usepackage{epstopdf}
\usepackage{ulem}
\usepackage{amsmath}
\usepackage{amsfonts}
\usepackage{url}
\usepackage{amssymb}
\usepackage{graphics}
\usepackage{epsfig}
\usepackage{verbatim}
\usepackage{amsthm}
\usepackage{geometry}
\usepackage{framed}    % to make framed boxes
\usepackage{scalefnt}
\usepackage{alltt}
\usepackage{hyperref}
\hypersetup{
    bookmarks=true,         % show bookmarks bar?
    unicode=false,          % non-Latin characters in Acrobat's bookmarks
    pdftoolbar=true,        % show Acrobat's toolbar?
    pdfmenubar=true,        % show Acrobat's menu?
    pdffitwindow=true,      % page fit to window when opened
    pdftitle={My title},    % title
    pdfauthor={Author},     % author
    pdfsubject={Subject},   % subject of the document
    pdfnewwindow=true,      % links in new window
    pdfkeywords={keywords}, % list of keywords
    colorlinks=true,        % false: boxed links; true: colored links
    linkcolor=blue,         % color of internal links
    citecolor=blue,        % color of links to bibliography
    filecolor=magenta,      % color of file links
    urlcolor=cyan           % color of external links
}

\def\iid{\buildrel {\rm i.i.d.} \over \sim}

\def\eqd{\buildrel {\rm d} \over =}

\def\i.i.d.{\buildrel {\rm i.i.d.} \over \sim}

\def\cw#1 { \overset{\mathbb{P}}{\underset{#1}{\longrightarrow}} }
\def\Real{\mathbb{R}}
\def\Natu0{\mathbb{N}_0}

\def\R{{\mathbb R}}  %%
\def\N{{\mathbb N}}  %%
\def\EE{{\mathbb E}}  %
\def\V{{\mathbb V}}  %%

\def\E#1{{\mathbb E}\left[#1\right]}

\def\Var#1{{\mathbb V}\left(#1\right)}

\def \rcov#1#2 {{\rm cov}_{#1}\left( #2\right)}

\newtheorem{lemma}{Lemma}[section]
\newtheorem{theorem}[lemma]{Theorem}

\newtheorem{corollary}[lemma]{Corollary}
\newtheorem{remark}[lemma]{Remark}
\newtheorem{remarks}[lemma]{Remarks}
\newtheorem{proposition}[lemma]{Proposition}

\newtheorem{model}[lemma]{Model}

\newtheoremstyle{italicized}
  {3pt}   % Space above
  {3pt}   % Space below
  {\itshape} % Body font (italicized)
  {}      % Indent amount
  {\bfseries} % Theorem head font
  {.}     % Punctuation after theorem head
  {.5em}  % Space after theorem head
  {}

\theoremstyle{italicized}
\newtheorem*{model2}{Smith--Miller Model}

\newtheoremstyle{italicized}
  {3pt}   % Space above
  {3pt}   % Space below
  {\itshape} % Body font (italicized)
  {}      % Indent amount
  {\bfseries} % Theorem head font
  {.}     % Punctuation after theorem head
  {.5em}  % Space after theorem head
  {}

\theoremstyle{italicized}
\newtheorem*{model3}{Generalized Smith--Miller Model}

\numberwithin{equation}{section}

\begin{document}
	\begin{titlepage}	
		\thispagestyle{empty}

\title{An Observation-Driven State-Space Model for Claims
Size Modeling}

\author{Jae Youn Ahn\fnref{thirdfoot}\corref{bbb}}%\corref{bbb}
\author{\quad Himchan Jeong \fnref{firstfoot}\corref{bbb}}
\author{\quad Mario V.~W\"uthrich\fnref{fourthfoot}\corref{bbb}}

\cortext[bbb]{Corresponding authors/equal contribution}
\fntext[thirdfoot]{Department of Statistics, Ewha Womans University, Seoul, Republic of Korea. Email: \url{jaeyahn@ewha.ac.kr}}
\fntext[firstfoot]{Department of Statistics and Actuarial Science, Simon Fraser University, BC, Canada. Email: \url{himchan_jeong@sfu.ca}}
\fntext[fourthfoot]{RiskLab, Department of Mathematics, ETH Zurich, Switzerland.
Email: \url{mario.wuethrich@math.ethz.ch}
}
\begin{abstract}
State-space models are popular models in econometrics. Recently,
these models have gained some popularity in the actuarial literature. The best known state-space models are of Kalman-filter type. These models are so-called parameter-driven because the observations do not impact the state-space dynamics. A second
less well-known class of state-space models are so-called observation-driven state-space models where the state-space dynamics is also impacted by the actual observations. A typical example is the
Poisson-Gamma observation-driven state-space model for counts data. This Poisson-Gamma model is fully analytically tractable.
The goal of this paper is to develop a Gamma-Gamma observation-driven state-space model for claim size modeling. We provide fully tractable versions of Gamma-Gamma observation-driven state-space models, and these versions extend the work of
\cite{smith1986non} by allowing for a fully flexible variance behavior. Additionally, we demonstrate that the proposed model aligns with evolutionary
credibility, a methodology in insurance that dynamically adjusts premium rates over time using evolving data.
\end{abstract}

\maketitle



\end{titlepage}

\newpage

\section{Introduction and motivation}

Static random effect models \citep{laird1982random, lee1996hierarchical} have a long tradition in insurance ratemaking to account for heterogeneity among insurance policyholders
\citep{bichsel1964erfahrungs}. These models assume that the individual claim behavior is influenced by an unobserved latent variable, which remains constant over time for a given policyholder. By estimating this latent variable from past observations, static random effect models provide a systematic way to predict future claims based on past claims experience. This approach has proven effective in contexts where temporal dynamics are less critical, as it simplifies the modeling process and it requires fewer assumptions about the evolution of risk factors.

However, many in real-world applications it is crucial to capture temporal dynamics and evolving risks, necessitating an extension from static random effect models to dynamic frameworks. State-space models \citep{Kalman1960, anderson2005optimal} naturally comply with this need by introducing latent processes that evolve over time, allowing for the modeling of longitudinal data with temporal dependencies. These models are particularly suited for applications in insurance, where the evolution of risk factors plays a critical role \citep{pinquet2001allowance, bolance2003time}. Parameter-driven state-space models, which rely on stochastic latent processes, provide a robust framework for such settings.
However, these models often lack analytical tractability, as closed-form expressions for filtering and likelihood evaluations are generally unavailable except in limited cases, such as under Gaussian assumptions. As a result, numerical methods are often employed, which increase computational complexity, require significant processing time and which my restrict explainability because sensitivity analyses may not easily be available. This makes practical implementation challenging in large-scale or real-time applications \citep{doucet2000sequential, arulampalam2002tutorial}.

Alternatively, credibility methods \citep{whitney1918, buhlmann1967experience, BuhlmannStraub} can be used as a simpler (approximative) approach to predict claims, and parameter-driven
state-space versions were studied in \citep{pinquet2020poisson, pinquet2020positivity, ahn2021order}. However, credibility methods typically only focus on providing the predictive mean, and they do not yield the full predictive distribution. Of course, this is a significant limitation of credibility methods, e.g., in risk management the full predictive
distributions needs to be studied.

To address these challenges, observation-driven state-space models \citep{smith1986non, harvey1989time} have proven to be a practical and effective alternative. These models maintain the dynamic structure of state-space frameworks while allowing for analytical solutions through the use of conjugate prior relationships. \cite{smith1986non} introduced an observation-driven state-space model for exponentially distributed responses, which was later applied to insurance ratemaking \citep{bolance2007greatest}. This model was subsequently extended to Gamma-distributed responses   using the Bayesian Gamma-Gamma conjugacy \citep{youn2023simple}, increasing its relevance and applicability in insurance modeling. Similarly, \cite{harvey1989time} introduced an observation-driven state-space model tailored for count data.

Despite their advantages, existing observation-driven state-space models \citep{smith1986non, harvey1989time} are limited in their ability to capture diverse variance behaviors. Specifically, these models typically assume an increasing variance process that
asymptotically goes to infinity. This restricts their flexibility in applications where the variance dynamics may vary, may be uniformly bounded or may even be decreasing. The limitations of the model of \cite{harvey1989time} were formally analyzed in \cite{ahn2023classification}, where flexible variance extensions were proposed to address these constraints. The model of \cite{smith1986non} has not yet been evaluated in this regard, leaving its ability to handle diverse variance behaviors unexplored. In this article, we address this gap by examining the limitations of the model of \cite{smith1986non} and providing a generalized framework that overcomes these issues.

To this end, we extend the model introduced in \cite{smith1986non} to develop a flexible observation-driven state-space framework specifically designed for continuous positive data, based on Gamma distributions. These generalized models are able to accommodate a full range of variance behaviors, including stationary, increasing, and decreasing patterns, thereby addressing the constraints of the original framework. By retaining the analytical tractability of the original observation-driven approach, the proposed models ensure closed-form solutions for filtering, estimation and predictive distributions, while significantly enhancing the flexibility in modeling the variance dynamics.
These features make these models particularly well-suited for applications involving continuous positive data, such as insurance claim sizes.
Furthermore, we show that the proposed model is consistent with evolutionary credibility, dynamically adjusting to evolving data for application in insurance pricing strategies.

This manuscript is structured as follows. We first revisit the Smith–Miller Model, paying particular attention to the variance behavior of its state-space process. To achieve full flexibility in variance behavior, we introduce the Generalized Smith–Miller Model in Section \ref{sec: Generalized Smith--Miller Model} and analyze its variance dynamics in Section \ref{sec.4}. Section \ref{sec: Forecasting and evolutionary credibility} covers model fitting, forecasting, and the evolutionary credibility formula within the Generalized Smith–Miller framework. Sections \ref{section simulation study} and \ref{section real study} present a simulation study and a real data analysis, respectively. Finally, concluding remarks are provided in Section \ref{Concluding remarks}.

%This manuscript is structured as follows.
%We first revisit the Smith--Miller Model and we pay
%special attention to the variance behavior of its state-space process.
%To allow for full flexibility in this variance behavior, we introduce
%a Generalized Smith--Miller Model in Section \ref{sec: Generalized Smith--Miller Model}, and in Section \ref{sec.4}, we analyze the resulting
%variance behaviors. In Section \ref{sec: Forecasting and evolutionary credibility}, we address model fitting, forecasting and we give the evolutionary credibility formula in the Generalized Smith--Miller Model.
%Sections \ref{section simulation study} and \ref{section real study} give a simulation study and a real data analysis. Finally, in Section \ref{Concluding remarks}, we conclude.

\section{Revisiting the Smith--Miller Model}

We start by revisiting the \textit{Smith--Miller Model}  \citep{smith1986non} which originally utilized the exponential distribution to characterize the observation distribution. This framework was later extended  \citep{youn2023simple} to include the more general Gamma distribution, allowing additionally for
over-dispersion. We emphasize the need to generalize this latter model by demonstrating that it can accommodate only a very limited variance behavior of the state-space dynamics.

We begin by introducing the notation. For finite sequences, let $Y_{1:t} = (Y_1, \ldots, Y_t)$, and interpret $Y_{1:0}$ as the empty sequence, which generates the trivial $\sigma$-field $\sigma(Y_{1:0}) = \{\emptyset, \Omega\}$ on the underlying probability space $(\Omega, \mathcal{F}, \mathbb{P})$.

We denote by $\Gamma(\alpha, \beta)$ the Gamma distribution
with shape parameter $\alpha > 0$ and scale parameter $\beta > 0$. It is positively supported and has probability density function
on $\R_+$
   \[
   f\left( x ; \alpha, \beta \right) =
     \frac{\beta^\alpha}{\Gamma(\alpha)}\, x^{\alpha-1}\exp\left(-\beta x\right) \qquad \text{ for $x>0$}.
       \]
The mean and variance are given by $\alpha/\beta$ and $\alpha/\beta^2$, respectively.

\subsection{Review of the Smith--Miller Model}

The Smith--Miller Model involves two stochastic processes: a latent state-space process $(\Theta_t)_{t \ge 1}$, which acts as a hidden driver of the dynamics, and an observable process of response variables $(Y_t)_{t \ge 1}$, which depends on the latent state-space process.
A key feature of the model is the incorporation of a feedback loop from the responses to the state-spaces, impacting the state-space updates, as see \eqref{stillgamma00}, below.
This feedback mechanism distinguishes observation-driven state-space models from their parameter-driven counterparts of \cite{Kalman1960} type.
Specifically, following the classification in \cite{cox1981statistical}, parameter-driven state-space models are defined by specifying the state-space updates
\[\text{from} \quad
\left.\Theta_t \right|_{\Theta_{1:t}} \quad \text{to} \quad \left.\Theta_{t+1} \right|_{\Theta_{1:t}}.
\]
In contrast, in observation-driven models, the state-space update is defined by specifying the state-space transitions
\[\text{from} \quad
\left.\Theta_t \right|_{Y_{1:t}} \quad \text{to} \quad \left.\Theta_{t+1} \right|_{Y_{1:t}}.
\]
This incorporates the observed responses into the state-space update.

\begin{model2}[\cite{smith1986non, youn2023simple}]
%\label{model.0}
Consider a fixed dispersion parameter $\psi>0$, initialization $a_{1|0}>1$, and two exogenous sequences
$(v_t)_{t\ge 1}\subseteq \N_0$ and $
(\mu_t)_{t\ge 1}\subseteq \R_+$.
The processes of the responses $(Y_t)_{t \ge 1}$ and  of
the state-spaces $(\Theta_t)_{t \ge 1}$ are defined as follows:

\begin{itemize}
\item[(0)] \textbf{Initialization at $t=1$}: Initialize the state-space process $(\Theta_t)_{t \ge 1}$ by
\begin{equation}
\label{eq1}
\left.\Theta_1 \right|_{Y_{1:0}} \sim \Gamma\left(1 + a_{1|0}, b_{1|0}\right),
\end{equation}
with $a_{1|0}=b_{1|0}$ giving unit mean of the inverse
initial state-space, $\EE[\Theta_1^{-1}]=1$.

\item[(1)] \textbf{Observation equation for $t \ge 1$}: For given $Y_{1:t-1}$ and $\Theta_{1:t}$, the response $Y_t$ satisfies
\begin{equation}
\label{equation2}
\left. Y_t \right|_{Y_{1:t-1}, \Theta_{1:t}} \sim \Gamma\left(\frac{v_t}{\psi}, \frac{\Theta_t}{\mu_t \psi}\right), \quad \text{if } v_t > 0,
\end{equation}
and $Y_t \vert_{Y_{1:t-1}, \Theta_{1:t}} = 0$, almost surely, if $v_t = 0$.

\item[(2)] \textbf{State-space update from $t$ to $t+1$}: Assume the filtering distribution at time $t \ge 1$ is given by
\begin{equation}\label{filter 1}
\left.\Theta_t\right|_{Y_{1:t}} \sim \Gamma\left(1 + a_t, b_t\right),
\end{equation}
then, the predictive distribution for $\Theta_{t+1}$ is assumed to satisfy
\begin{equation}\label{stillgamma00}
\left.\Theta_{t+1}\right|_{Y_{1:t}} \sim \Gamma\left(1 + a_{t+1|t}, b_{t+1|t}\right),
\end{equation}
where $a_{t+1|t}$ and $b_{t+1|t}$ follow the updating rule
\[
  a_{t+1|t} = q_t a_t> 1 \qquad \text{ and } \qquad
  b_{t+1|t} = q_t b_t > 0,
\]
with $q_t$ being defined by
\begin{equation}\label{eq.2}
q_t = \frac{\gamma\left(a_t - 1\right) + 1}{a_t},
\end{equation}
for a constant parameter $\gamma \in (0,1]$.
\end{itemize}
\end{model2}

\begin{remarks}\normalfont
\begin{itemize}
\item
We show in \eqref{eq41}, below, that \eqref{filter 1}
is a straightforward implication of Bayes' rule. Together with
\eqref{eq.2} and the initial condition $a_{1|0}>1$, this implies
that $a_{t}>1$ as well as $a_{t|t-1}>1$, for all $t\ge 1$, so that the conditional variances of the inverse state-spaces are well-defined for all $t\ge 1$.
\item
We discuss the observation equation \eqref{equation2} and its application to insurance.
For $v_t\ge 1$, it provides
us with the first two conditional moments
\begin{equation}\label{Bayesian filtering Gamma}
\EE \left[\left. Y_{t} \right|Y_{1:t-1}, \Theta_{1:t}\right]
= \frac{v_t\mu_t}{\Theta_t} \qquad \text{ and } \qquad
\Var{\left. Y_{t} \right|Y_{1:t-1}, \Theta_{1:t}}
= v_t\psi\,\frac{\mu^2_t}{\Theta^2_t}.
\end{equation}
This Gamma response $Y_{t}$ is in its additive form;
see \cite{jorgensen1997theory} for the difference between reproductive and additive forms.
For a non-unit exposure\footnote{The models in \cite{smith1986non} and \cite{youn2023simple} do not include the exposure $v_t$ in the observation equation \eqref{equation2}, but for insurance modeling this is very convenient.} $v_t\ge 2$, it can be
interpreted as the sum of $v_t$ conditionally i.i.d.~gamma distributed
individual claims, that is,
\[
\left. Z_{t, 1}, \ldots, Z_{t, v_t} \right\vert_{\Theta_t} \iid \Gamma(1/\psi, \Theta_t/(\mu_t \psi)),
\]
since we
have the distributional identity
\begin{equation}\label{eqn:distid}
\left. Y_{t}\right|_{Y_{1:t-1}, \Theta_{1:t}}~ \stackrel{\rm (d)}{=} ~ \sum_{j=1}^{v_t} \left. Z_{t,j} \right|_{\Theta_t}
\sim
\Gamma\left( \frac{v_t}{\psi}, \frac{\Theta_{t} }{\mu_{t}\psi}\right).
\end{equation}
That is, $Y_t$ can be interpreted as the aggregate claim amount in the given (time) period $t$ of i.i.d.~individual Gamma distributed claims $\left(Z_{t,1}, \ldots, Z_{t, v_t}\right)$, and it is sufficient to
know the aggregate (or average) claim amount $Y_t$ for each
time period $t$, because it forms a sufficient statistics for parameter estimation in this problem.
\end{itemize}
\end{remarks}

The state-space update under Smith--Miller Model gives is a mean-stationary
dynamics of the inverse state-space process $( \Theta_t^{-1})_{t\ge 1}$
when passing from filtering to predictive distribution, i.e.,
\begin{equation}\label{eq.a1}
\E{\left.\Theta_{t+1}^{-1}\right\vert Y_{1:t}} = \E{\left.\Theta_{t}^{-1}\right\vert Y_{1:t}}.
\end{equation}
This in turn implies unconditional mean stationarity of the inverse state-spaces. On the other hand, the conditional variances of the inverse state-spaces are proportionally increasing
\begin{equation}\label{eq.a2}
\Var{\left.\Theta_{t+1}^{-1}\right\vert Y_{1:t}} = \frac{1}{\gamma}\,\Var{\left.\Theta_{t}^{-1}\right\vert Y_{1:t}}.
\end{equation}
Remark, that in dealing with the ``stationarity of the state-spaces'', we are mainly interested in the mean and variance stationarity of the inverse state-space process $(\Theta^{-1}_t)_{t \ge 1}$ rather than the process $(\Theta_t)_{t \ge 1}$ itself. This distinction is motivated by the conditional mean and variance in \eqref{Bayesian filtering Gamma}, which highlights the primary interest in the mean and variance behavior of $(\Theta^{-1}_t)_{t \ge 1}$.

The popularity of the Smith--Miller Model is explained by the fact
that it allows for explicit formulations of the predictive means
$\EE[Y_{t+1}| Y_{1:t}]$ and the probability density functions of the
response variables $Y_{1:t}$ in a recursive form \citep{youn2023simple}.

%%%%%%%%%%%%%%%%%%%%%%%%

\subsection{Variance behavior of the inverse state-space dynamics}
The Smith--Miller Model is widely used in applications, including insurance, due to its intuitive interpretation of the state-space updates, see \eqref{eq.a1} and \eqref{eq.a2}, and its property of a closed-form expression of the likelihood function. However, the Smith--Miller Model is not fully flexible, as the variance behavior of the inverse state-spaces is increasing, as the following result shows.

\begin{proposition}\label{theo.10}
Consider the setting in the Smith--Miller Model. Then, we have
\[
\Var{\Theta^{-1}_{t}} \le
\Var{\E{\Theta^{-1}_{t} \,|\, Y_{1:t}}}
+
\frac{1}{\gamma}\,\EE\left[\Var{\Theta^{-1}_{t} \,|\, Y_{1:t}}\right]
=
\Var{\Theta^{-1}_{t+1}}
<\infty,\quad t\ge 1
\]
with a strict inequality for $\gamma \in (0,1)$.
\end{proposition}

\begin{proof}[Proof of Proposition \ref{theo.10}]
For $\gamma \in (0,1]$, the tower property for conditional expectations
gives us
\begin{eqnarray*}
\Var{\Theta^{-1}_{t}}
&=&\Var{\E{\Theta^{-1}_{t} \,|\, Y_{1:t}}}
+
\EE\left[\Var{\Theta^{-1}_{t} \,|\, Y_{1:t}}\right]
\\&\le&
\Var{\E{\Theta^{-1}_{t} \,|\, Y_{1:t}}}
+
\frac{1}{\gamma}\,\EE\left[\Var{\Theta^{-1}_{t} \,|\, Y_{1:t}}\right]
\\&=&
\Var{\E{\Theta^{-1}_{t+1} \,|\, Y_{1:t}}}
+
\EE\left[\Var{\Theta^{-1}_{t+1} \,|\, Y_{1:t}}\right]
=
\Var{\Theta^{-1}_{t+1}},
\end{eqnarray*}
where the last equality follows from \eqref{eq.a1} and \eqref{eq.a2}. Furthermore, since $\E{\Var{\Theta_t^{-1}\,|\, Y_{1:t}}}>0$, the equality becomes strict for $\gamma\in(0,1)$.
This completes the proof.
\end{proof}

\section{Generalized Smith--Miller Model}
\label{sec: Generalized Smith--Miller Model}

Proposition \ref{theo.10} demonstrates that the state-space updates \eqref{stillgamma00} cannot lead to a decreasing or stationary variance
behavior of the inverse state-space dynamics.
We introduce a generalized version of the Smith--Miller Model. Modifying the state-space update \eqref{stillgamma00} appropriately will allow for flexible variance behaviors, and their detailed analysis is provided in Section \ref{sec.4}, below.
The following describes the generalized version of the Smith--Miller Model.

\begin{model3}
Consider a fixed dispersion parameter $\psi>0$, initialization $a_{1|0} > 0$, and three exogenous sequences
\begin{equation*}
\left(\xi_t\right)_{t\ge 1} \subseteq \Xi, \quad
\left(v_t\right)_{t\ge 1} \subseteq \mathbb{N}_0, \quad \text{and} \quad \left(\mu_t\right)_{t\ge 1} \subseteq \mathbb{R}_+,
\end{equation*}
where $\Xi \subseteq \mathbb{R}^K$ is a convex parameter space for a positive integer $K$.
The processes of the responses $(Y_t)_{t \ge 1}$ and the state-spaces $(\Theta_t)_{t \ge 1}$ are assumed to have the same initialization
\eqref{eq1}, $t=1$, and the same observation equation \eqref{equation2}, $t\ge 1$, as the Smith--Miller Model.
The state-space update from time $t$ to time $t+1$
is defined as follows:

\begin{itemize}
\item[(2)] \textbf{State-space update from $t$ to $t+1$}: Assume the filtering distribution at time $t \ge 1$ is given by
\begin{equation}
\label{eq.3400}
\left.\Theta_t\right|_{Y_{1:t}} \sim \Gamma\left(1 + a_t, b_t\right),
\end{equation}
then, the predictive distribution for $\Theta_{t+1}$ is assumed to satisfy
\begin{equation}
\label{stillgamma0}
\left.\Theta_{t+1}\right|_{Y_{1:t}} \sim \Gamma\left(1 + a_{t+1|t}, b_{t+1|t}\right),
\end{equation}
where $a_{t+1|t}$ and $b_{t+1|t}$ follow the recursion
\begin{equation}
\label{eq.update}
%\begin{aligned}
  a_{t+1|t} = A(a_t, b_t; \xi_t) > 0
  \qquad \text{ and } \qquad
  b_{t+1|t} = B(a_t, b_t; \xi_t) > 0,
%\end{aligned}
\end{equation}
for given measurable functions $A: \mathbb{R}^2_+ \times \Xi \to \mathbb{R}_+$ and $B: \mathbb{R}^2_+ \times \Xi \to \mathbb{R}_+$ specified below.
\end{itemize}
\end{model3}

With suitable choices for the update functions $A$ and $B$ in \eqref{eq.update}, the model results in fully tractable models with mean stationary state-space processes. This is detailed in the next subsection.
Before proceeding to this, we briefly address the well-definedness of the Generalized Smith--Miller Model.
Although we introduced formula \eqref{eq.3400} as an assumption, it is, in fact, a direct consequence of the conditional response distribution \eqref{equation2} and the latent distribution \eqref{stillgamma0}.
This result is obtained by applying a Bayesian inference step, often referred to as \textit{filtering} in the time series literature, within a Gamma-Gamma conjugate prior framework.
To clarify this, we briefly explain the process, generically denoting a density by $f$. By applying Bayes' formula
\begin{eqnarray*}
f\left(\left.\Theta_t \right|Y_{1:t}\right)
&\propto&
f\left(\left.Y_t \right|Y_{1:t-1}, \Theta_t\right)
f\left(\left.\Theta_t \right|Y_{1:t-1}\right),
\end{eqnarray*}
and substituting the explicit forms of the Gamma densities, the filtering distribution becomes a Gamma distribution with the following parameter updates
\begin{equation}\label{eq41}
    a_t := a_{t|t-1} + \frac{v_t}{\psi}>0
    \qquad\hbox{and}\qquad
    b_t := b_{t|t-1} + \frac{Y_t}{\mu_t \psi}>0.
\end{equation}
In this sense, formula \eqref{eq.3400} is not a model assumption but rather a (mathematical) consequence of the state-space update in \eqref{stillgamma0}, the conditional response distribution \eqref{equation2}, and the initial condition \eqref{eq1}. As a result, the Generalized Smith--Miller Model is well-defined.

\bigskip

%\subsection{Linear evolutionary model with stationary mean}

We discuss the state-space update \eqref{stillgamma0}, and we specify explicit functional forms of $A$ and $B$ in \eqref{eq.update} that result in linearly tractable models with a {\it mean stationary}
inverse state-space dynamics.
Among various possible choices of the update functions $A$ and $B$, we consider the following positive affine functions for $\xi_{t}=( \xi_{1,t}, \ldots, \xi_{6,t})^\top\in \Xi$
\[
a_{t+1 | t} = A(a_t, b_t;\xi_{t}) = \xi_{1,t} +  \xi_{2,t} a_t + \xi_{3,t} b_t>0,
\]
and
\[
b_{t+1 | t} = B(a_t, b_t;\xi_{t}) = \xi_{4,t} +  \xi_{5,t} a_t + \xi_{6,t} b_t>0.
\]
If, for some $\Delta_t\in(0,1]$, we furthermore require the following thinning property
\begin{equation}\label{thinning a}
\E{\Theta_{t+1}^{-1}\,|\, Y_{1:t}} = \Delta_t \E{\Theta_{t}^{-1}\,|\, Y_{1:t}} +\left( 1-\Delta_t\right), \qquad \text{ for all $a_t, b_t>0$,}
\end{equation}
then this implies the restriction for $t\ge 1$
\[
\xi_{1,t}=\xi_{3,t}= \xi_{4,t}=0 \qquad\hbox{and}\qquad \xi_{6,t}=\Delta_t \xi_{2,t}.
\]
The thinning in
\eqref{thinning a} is motivated by the intuition that we would
like to carry forward some past information $Y_{1:t}$, reflected
by the first term in \eqref{thinning a}, but we also want to
add some new noise to the state-space process reflected
by the second term $(1-\Delta_t)\cdot 1$, where this can be
interpreted as the prior mean part $\EE[\Theta^{-1}_t]=1$.
The validity of this prior mean is implied
by mean stationarity, which still needs to be established,
see Proposition \ref{lemma 1 mean stationarity}, below.
After appropriate reparametrization, we present the following linear evolutionary model as a specific instance of the Generalized Smith–Miller Model.
\begin{model}
\label{model.2}
Consider a fixed dispersion parameter $\psi>0$, initialization $a_{1|0} > 0$, and four exogenous sequences
\begin{equation*}
(p_t)_{t\ge 1}\subseteq [0,1], \quad (q_t)_{t\ge 1}\subseteq \Real_+, \quad
(v_t)_{t\ge 1}\subseteq \N_0 \quad
\hbox{and}\quad (\mu_t)_{t\ge 1}\subseteq \R_+.
\end{equation*}
Consider the processes of the responses  $\left(Y_t\right)_{t \ge 1}$ and the state-spaces $\left(\Theta_t\right)_{t\ge 1}$ as in the Generalized Smith--Miller Model, where the general recursions \eqref{eq.update} are replaced by, for $t \ge 1$,
\begin{equation}\label{eq.update22}
 a_{t+1|t} = (p_t+q_t)a_t\qquad\hbox{and}\qquad
 b_{t+1|t} = p_ta_t + q_tb_t.
\end{equation}
\end{model}
Model \ref{model.2} is well-defined if $a_{t+1|t}>0$ and $b_{t+1|t}>0$ for all $t\ge 0$, which
is an immediate consequence of recursions \eqref{eq41}
and \eqref{eq.update22} as well as the initialization $a_{1|0}
=b_{1|0}>0$.
Simple algebraic computations show that $\Delta_t$ in \eqref{thinning a} can be represented as
\begin{equation}\label{credibility Delta}
\Delta_t=q_t/(p_t+q_t)~\in~(0,1].
\end{equation}
In view of \eqref{eq41} along with the updates in \eqref{eq.update22}, the observations
$(Y_t)_{t\ge 1}$ only enter the sequence $(b_{t+1|t})_{t\ge 1}$,
and, henceforth, $(a_{t+1|t})_{t\ge 1}$ is deterministic.

\begin{remark}\normalfont
Parametric representations of the state-space updates, such as the one from \eqref{eq.3400} to \eqref{stillgamma0}, have been highlighted as key steps in the related literature \citep{smith1986non, harvey1989time, youn2023simple, ahn2023classification}. These representations are often presented as tools for simulating the state-space process and providing intuitive interpretations of the model dynamics. We could present similar
parametric representations for Model \ref{model.2}, however, it is important to realize that such parametric representations are not
unique. This non-uniqueness diminishes the interpretive value of these constructions, as different parametrizations can result in the same underlying dynamics.
\end{remark}

\begin{proposition} \label{lemma 1 mean stationarity}
Under the setting of Model \ref{model.2}, we have mean stationarity $\EE[\Theta_t^{-1}]=1$,
$t\ge 1$.
\end{proposition}

\begin{proof}[Proof of Proposition \ref{lemma 1 mean stationarity}]
We compute
\begin{eqnarray}\nonumber
\EE \left[\left.\Theta^{-1}_{t+1}\right| Y_{1:t}\right]
&=& \frac{b_{t+1|t}}{a_{t+1|t}}
~=~ \frac{q_tb_t+p_ta_t}{(p_t+q_t)a_t}
~=~ \frac{q_t}{p_t+q_t}\frac{b_t}{a_t}
+ \frac{p_t}{p_t+q_t}
\\\label{conditional means 33}
&=&  \Delta_t\EE \left[\left.\Theta^{-1}_{t}\right| Y_{1:t}\right]
+ (1-\Delta_t).
\end{eqnarray}
Using the tower property of conditional
expectations, this implies
\begin{equation*}
\EE \left[\left.\Theta^{-1}_{t+1}\right| Y_{1:t-1}\right]
=
\EE \left[\left. \EE \left[\left.\Theta^{-1}_{t+1}\right| Y_{1:t}\right]
\right| Y_{1:t-1}\right]
= \Delta_t\EE \left[\left.\Theta^{-1}_{t}\right| Y_{1:t-1}\right]
+ (1-\Delta_t).
\end{equation*}
By induction we receive the claim using initialization $b_{1|0}=a_{1|0}>0$.
\end{proof}

\section{Variance behavior of the state-space process}\label{sec.4}

In Proposition \ref{lemma 1 mean stationarity}, we proved
that the inverse state-space process $(\Theta^{-1}_t)_{t\ge 1}$ is mean stationary under Model \ref{model.2}.
The goal of this section
is to analyze its long-term variance behavior
under different specifications of the exogenous sequences $(p_t)_{t \ge 1}\subseteq [0,1]$ and $(q_t)_{t \ge 1}\subseteq \R_+$.

We start with a general result, Lemma \ref{cor.2}, which is the basic tool to analyze the variance behaviors in the subsequent models and subsections:
in Section \ref{subsection stationary}, we give
parameter sequences that lead to a variance stationary inverse
state-space process; in Section \ref{subsection_increasing},
we provide an example with an increasing variance behavior that
can asymptotically explode (diverge); and in Section
\ref{subsection decreasing}, we provide an example that
has a decreasing variance behavior, in the extreme case
converging to zero. That is, asymptotically there is no
randomness coming from the state-space process.

%Moreover, in \ref{app.1}, we provide a detailed characterization of the asymptotic variance behavior for a homogeneous version of Model \ref{model.2}, i.e., setting constant parameter sequences
%\begin{equation}\label{homogeneous model} \mu_t \equiv \mu > 0, \qquad v_t \equiv 1, \qquad p_t \equiv p \in [0,1], \qquad \text{and} \qquad q_t \equiv q > 0. \end{equation}
%This analysis further demonstrates the full flexibility of the variance behavior of the state-space process under Model \ref{model.2}, allowing the model to adapt to various patterns observed in real data through appropriate parameter choices.

\subsection{A general result on the variance of the inverse state-space process}
We start by a general recursive formula for the variance
of the inverse state-space process.

\begin{lemma}\label{cor.2}
Consider the setting of Model \ref{model.2}, and assume that the following variances are finite
\[
\Var{\Theta^{-1}_{t}}<\infty, \qquad \text{ for all $t\ge 1$.}
\]
The variances of the inverse state-space process satisfy the following recursion
\begin{equation}\label{general variance formula}
\Var{\Theta^{-1}_{t+1}}
=\frac{q_t^2}{p_t+q_t}\,
\frac{a_t-1}{(p_t+q_t)a_t-1}
\Var{\Theta^{-1}_{t}}
+\left(1-\frac{q_t^2}{p_t+q_t}\right)\frac{1}{(p_t+q_t)a_t-1}.
\end{equation}
\end{lemma}

\begin{proof}[Proof of Lemma \ref{cor.2}]
We start by proving the following three results
\begin{eqnarray}\label{result theorem 1.a}
\Var{\Theta^{-1}_{t+1}}
&=& \frac{a_{t+1|t}}{a_{t+1|t}-1}
\Var{\EE\left[\left.\Theta^{-1}_{t+1}\right| Y_{1:t}\right]
}
+\frac{1}{a_{t+1|t}-1},
\\\label{result theorem 1.b}
\Var{\Theta^{-1}_{t}}
&=& \frac{a_{t}}{a_{t}-1}
\Var{\EE\left[\left.\Theta^{-1}_{t}\right| Y_{1:t}\right]
}
+\frac{1}{a_{t}-1},
\\\label{result theorem 1.c}
\Var{\EE\left[\left.\Theta^{-1}_{t+1}\right| Y_{1:t}\right]
}
&=& \Delta_t^2\Var{\EE\left[\left.\Theta^{-1}_{t}\right| Y_{1:t}\right]
}.
\end{eqnarray}
We use the tower property and mean stationarity, see Proposition
\ref{lemma 1 mean stationarity}, to receive the first claim \eqref{result theorem 1.a} as follows
\begin{eqnarray*}
\Var{\Theta^{-1}_{t+1}}
&=&
\EE\left[\Var{\left.\Theta^{-1}_{t+1}\right| Y_{1:t}}\right]
+\Var{\EE\left[\left.\Theta^{-1}_{t+1}\right| Y_{1:t}\right]}
\\&=&
\EE\left[\frac{b^2_{t+1|t}}{a^2_{t+1|t}(a_{t+1|t}-1)}
\right]
+\Var{\EE\left[\left.\Theta^{-1}_{t+1}\right| Y_{1:t}\right]
}
\\&=& \frac{1}{a_{t+1|t}-1}
\EE\left[\EE\left[\left.\Theta^{-1}_{t+1}\right| Y_{1:t}\right]^2
\right]
+\Var{\EE\left[\left.\Theta^{-1}_{t+1}\right| Y_{1:t}\right]
}
\\&=& \frac{a_{t+1|t}}{a_{t+1|t}-1}
\Var{\EE\left[\left.\Theta^{-1}_{t+1}\right| Y_{1:t}\right]
}
+\frac{1}{a_{t+1|t}-1}.
\end{eqnarray*}
Similarly we have we have for the second claim \eqref{result theorem 1.b}
\begin{eqnarray*}
\Var{\Theta^{-1}_{t}}
&=&
\EE\left[\Var{\left.\Theta^{-1}_{t}\right| Y_{1:t}}\right]
+\Var{\EE\left[\left.\Theta^{-1}_{t}\right| Y_{1:t}\right]}
\\&=&
\EE\left[\frac{b^2_{t}}{a^2_{t}(a_{t}-1)}
\right]
+\Var{\EE\left[\left.\Theta^{-1}_{t}\right| Y_{1:t}\right]
}
\\&=& \frac{1}{a_{t}-1}
\EE\left[\EE\left[\left.\Theta^{-1}_{t}\right| Y_{1:t}\right]^2
\right]
+\Var{\EE\left[\left.\Theta^{-1}_{t}\right| Y_{1:t}\right]
}
\\&=& \frac{a_{t}}{a_{t}-1}
\Var{\EE\left[\left.\Theta^{-1}_{t}\right| Y_{1:t}\right]
}
+\frac{1}{a_{t}-1}.
\end{eqnarray*}
Finally, we have for the third claim \eqref{result theorem 1.c}
\begin{equation*}
\Var{\EE\left[\left.\Theta^{-1}_{t+1}\right| Y_{1:t}\right]
}
= \Var{\Delta_t\EE \left[\left.\Theta^{-1}_{t}\right| Y_{1:t}\right]
+ (1-\Delta_t)
}
= \Delta_t^2\Var{\EE\left[\left.\Theta^{-1}_{t}\right| Y_{1:t}\right]
}.
\end{equation*}
The lemma is now proved by merging
\eqref{result theorem 1.a}-\eqref{result theorem 1.c}
as follows
\begin{eqnarray*}
\Var{\Theta^{-1}_{t+1}}
&=& \frac{(p_t+q_t)a_t}{(p_t+q_t)a_t-1}
\Var{\EE\left[\left.\Theta^{-1}_{t+1}\right| Y_{1:t}\right]
}
+\frac{1}{(p_t+q_t)a_t-1}
\\
&=&
\Delta_t^2 \frac{(p_t+q_t)a_t}{(p_t+q_t)a_t-1}
\Var{\EE\left[\left.\Theta^{-1}_{t}\right| Y_{1:t}\right]
}
+\frac{1}{(p_t+q_t)a_t-1}
\\
&=&
\Delta_t^2 \frac{(p_t+q_t)a_t}{(p_t+q_t)a_t-1}
\left(
\frac{a_{t}-1}{a_{t}}
\Var{\Theta^{-1}_{t}}
-\frac{1}{a_{t}}
\right)
+\frac{1}{(p_t+q_t)a_t-1}
\\
&=&
\Delta_t^2 \frac{(p_t+q_t)(a_t-1)}{(p_t+q_t)a_t-1}
\Var{\Theta^{-1}_{t}}
+\frac{1-\Delta_t^2(p_t+q_t)}{(p_t+q_t)a_t-1}.
\\
&=&\frac{q_t^2}{p_t+q_t}\,
\frac{a_t-1}{(p_t+q_t)a_t-1}
\Var{\Theta^{-1}_{t}}
+\left(1-\frac{q_t^2}{p_t+q_t}\right)\frac{1}{(p_t+q_t)a_t-1}.
\end{eqnarray*}
This completes the proof.
\end{proof}

\subsection{Model with a stationary variance process}
\label{subsection stationary}

We start with the stationary variance process case under Model \ref{model.2}. Recall the weights
$\Delta_t=q_t/(p_t+q_t) \in (0,1]$. Moreover, the sequence $(a_t)_{t\ge 1}$
is deterministic, i.e., does not depend on the responses $(Y_t)_{t\ge 1}$, these only enter the
sequence $(b_t)_{t\ge 1}$.

\begin{lemma}\label{lem.4}
Consider the setting of Model \ref{model.2} with the further condition $a_{1|0} > 1$, and assume that the variances of the inverse state-space process are finite
\[
\Var{\Theta^{-1}_{t}} < \infty, \qquad \text{ for all $t \ge 1$.}
\]
For a constant variance process $(\V(\Theta^{-1}_t))_{t \ge 1}$, the sequences $(p_t)_{t \ge 1}$ and $(q_t)_{t \ge 1}$ need to satisfy  for all $t \ge 1$
\begin{equation}\label{eq.112}
q_t = \frac{\Delta_t a_{1|0}}{a_t - \Delta_t^2 a_t + \Delta_t^2 a_{1|0}}.
\end{equation}
\end{lemma}

\begin{proof}[Proof of Lemma \ref{lem.4}]
Lemma \ref{cor.2} implies that for a constant variance process $( \V(\Theta^{-1}_t))_{t\ge 1}$ we need to have for all $t\ge 1$
the identity
\begin{eqnarray*}
\frac{1}{a_{1|0}-1}
&=&\frac{q_t^2}{p_t+q_t}\,
\frac{a_t-1}{(p_t+q_t)a_t-1}
\left(\frac{1}{a_{1|0}-1}\right)
+\left(1-\frac{q_t^2}{p_t+q_t}\right)\frac{1}{(p_t+q_t)a_t-1}
\\
&=&q_t\Delta_t\,
\frac{a_t-1}{q_ta_t/\Delta_t-1}
\left(\frac{1}{a_{1|0}-1}\right)
+\left(1-q_t\Delta_t\right)\frac{1}{q_t a_t/\Delta_t-1},
\end{eqnarray*}
this uses initialization $\V(\Theta^{-1}_1)=1/(a_{1|0}-1)$, $a_{1|0} > 1$.
The above identity is equivalent
to
\begin{equation*}
q_t a_t/\Delta_t-1
=q_t\Delta_t\,
(a_t-1)
+\left(1-q_t\Delta_t\right)(a_{1|0}-1)
=q_t\Delta_t\,
(a_t-a_{1|0})
+a_{1|0}-1 >0.
\end{equation*}
Solving this for $q_t$ gives the result.
\end{proof}

Lemma \ref{lem.4} is void unless there exists a Model \ref{model.2} where the corresponding variance process $(\V(\Theta^{-1}_t))_{t \ge 1}$ is constant over time $t \ge 1$. Motivated by Lemma \ref{lem.4}, we present the following model as an example of Model \ref{model.2} possessing the constant variance process property. The subsequent result demonstrates that this model exhibits the variance stationarity property.

\begin{model} \label{model.4}
With the additional exogenous sequence
\[
\left( \Delta_t\right)_{\ge 1}\subseteq (0,1],
\]
consider the setting in Model \ref{model.2} with further conditions $a_{1|0}>1$, and where the sequence $\left( q_t\right)_{\ge 1}\subseteq \Real_+$ is defined by \eqref{eq.112}.

\end{model}

For Model \ref{model.4} to be well-defined, we need to show $\left( q_t\right)_{\ge 1}\subseteq \Real_+$. %The following result, show that Model \ref{model.4} is well-defined.
For this, we show recursively $q_t>0$ and $a_t>0$,  supposed that $q_{t-1}>0$ and $a_{t-1}>0$.
The second sequence is initialized by $a_1\ge a_{1|0}>1$, and the first
one by
\begin{equation*}
q_1 = \frac{\Delta_1 a_{1|0}}{a_1(1-\Delta_1^2) +\Delta_1^2 a_{1|0}}>0.
\end{equation*}
The recursive step for $a_t$ follows from
\[
a_t\ge a_{t|t-1}=(p_{t-1}+q_{t-1})a_{t-1}=q_{t-1} \Delta_{t-1}^{-1}a_{t-1}>0,
\]
for $\Delta_{t-1} \in (0,1]$.
The recursive step
for $q_t$ is given by
\begin{equation*}
q_t = \frac{\Delta_t a_{1|0}}{a_t(1-\Delta_t^2) +\Delta_t^2 a_{1|0}}>0,
\end{equation*}
for $\Delta_{t} \in (0,1]$.
This proves that $(q_t)_{t\ge 1}\subseteq \R_+$ and
$(a_t)_{t\ge 1}\subseteq \R_+$.

\begin{theorem}\label{model.4 well defined}
Under the framework of Model \ref{model.4}, the inverse state-space process $(\Theta^{-1}_t)_{t \ge 1}$ is stationary in both its mean and variance.
\end{theorem}

\begin{proof}[Proof of Theorem \ref{model.4 well defined}]
By Proposition \ref{lemma 1 mean stationarity}, we have the
mean stationarity of $(\Theta_t^{-1})_{t\ge 1}$.
There remains to prove that this model
has a finite variance process. For this we require that
$a_t > 1$ and $a_{t|t-1}>1$ for all $t\ge 1$. The former
follows from the latter because $a_t\ge a_{t|t-1}$. So there remains to
prove the latter. We initialize
$a_{1|0} >1$. This implies $a_1\ge a_{1|0} >1$. From this
we can see that for $t=2$
\begin{equation*}
a_{2|1}=q_{1} \Delta_{1}^{-1}a_{1}
=\frac{a_{1|0}}{a_{1}-\Delta_{1}^2 a_{1} +\Delta_{1}^2 a_{1|0}} a_{1}
=\frac{a_{1|0}}{a_{1}-\Delta_{1}^2 (a_{1}-a_{1|0})} a_{1}
>a_{1|0}>1.
\end{equation*}
This implies $a_2\ge a_{2|1}>  a_{1|0}> 1$. Recursive iteration for $t\ge 3$
proves the claim.
\end{proof}

\subsection{Model with an increasing variance process}
\label{subsection_increasing}

Proposition \ref{theo.10} proves that the Smith--Miller Model exhibits an increasing variance behavior in the inverse state-space dynamics.
Consequently, Model \ref{model.2} is capable of accommodating such a behavior, as the Smith--Miller Model is a special case of Model \ref{model.2}.
That is, the Smith--Miller Model corresponds to Model \ref{model.2} under the following parameter settings, for $t\ge 1$,
\[
q_t = \frac{\gamma\left(a_t - 1\right) + 1}{a_t}
\qquad \text{and} \qquad
p_t = 0.
\]
We extent this result by showing that appropriate exogenous sequences in Model \ref{model.2}, the variance not only increases but, even diverges.

\begin{theorem}\label{corollary unbounded}
Consider Model \ref{model.2} with further conditions $p_t=0$ and $v_t=1$ for all $t\ge 1$,
and dispersion parameter $\psi=1$. Moreover, assume $0<q_t<\delta <1$ for all $t\ge 1$. Then, we have
\begin{equation}\label{eq.a3}
\liminf_{t \to \infty}\Var{\Theta^{-1}_{t}}=\infty.
\end{equation}
\end{theorem}

\begin{proof}[Proof of Theorem \ref{corollary unbounded}.]
Assume that there exists $N>0$ such that $a_t>1$ for all $t\ge N$,
otherwise \eqref{eq.a3} holds because the variance can only be
finite if $a_t >1$. That is, in this case we have $\Var{\Theta_t^{-1}}<\infty$ for all $t\ge N$.
Then, we come back to Lemma \ref{cor.2}.
Setting $p_t=0$, $t\ge 1$, we have from \eqref{general variance formula}
\begin{equation}\label{first and second terms}
\Var{\Theta^{-1}_{t+1}}
=
\frac{q_ta_t-q_t}{q_t a_t-1}\,
\Var{\Theta^{-1}_{t}}
+\left(1-q_t\right)\frac{1}{q_t a_t-1},
\end{equation}
for all $t\ge N$.
Consider the sequence $(a_t)_{t\ge 1}$
\begin{equation*}
a_{t+1}=a_{t+1|t}+v_{t+1}/\psi=q_{t}a_{t}+1
\le \delta a_{t} +1
\le \delta^t a_{1|0} + \sum_{k=0}^{t} \delta ^k
\le a_{1|0} + \frac{\delta}{1-\delta} < \infty,
\end{equation*}
for all $t \ge N$. This implies that the second term on the right-hand side of \eqref{first and second terms} is bounded uniformly from below
by a strictly positive constant. Moreover, the first ratio on the right-hand
side of \eqref{first and second terms} is bounded from below by 1, therefore,
\eqref{first and second terms} implies that a sequence $\left(\Var{\Theta^{-1}_{t}}\right)_{t\ge 1}$ is diverging for $t\to \infty$.
This completes the proof.
\end{proof}

\subsection{Model with a decreasing variance process}
\label{subsection decreasing}
Next, we present a setting of Model \ref{model.2}
that has decreasing variances of the inverse state-space process.
In the extreme case, these variances converge to zero, see
Corollary \ref{corollary limit variance zero}, below.
For this, we consider Model \ref{model.2} under the choice
\begin{equation}\label{eq.4}
p_t+q_t=1 \qquad \hbox{for all}\quad t\ge 1,
\end{equation}
with $(p_t)_{t\ge 1}\subset[0,1)$.

\begin{theorem}\label{theorem decreasing}
Consider Model \ref{model.2} with initialization $a_{1|0}>1$ and
with constraint \eqref{eq.4}. Then,
  \[
  \Var{\Theta^{-1}_{t+1}} \le \Var{\Theta^{-1}_{t}},
  \qquad \text{ for $t\ge 1$,}
  \]
  where the equality holds if and only if $p_t=0$.
\end{theorem}

\begin{proof}[Proof of Theorem \ref{theorem decreasing}.]
First we note that $\Delta_t \in (0,1]$ is deterministic, and $a_{t}, a_{t|t-1}>0$ for all $t\ge 1$.
Formula \eqref{conditional means 33} then implies
\[
\Var{\EE[\Theta^{-1}_{t+1} \,|\, Y_{1:t}]}=\Delta_t^2\,\Var{\EE[\Theta^{-1}_{t} \,|\, Y_{1:t}]}
\le \Var{\EE[\Theta^{-1}_{t} \,|\, Y_{1:t}]},
\]
and we have an equality if and only if $\Delta_t=1$ which is equivalent
to $p_t=0$.
Next, observe
\begin{eqnarray*}
    \EE\left[\Var{\Theta^{-1}_{t+1} \,|\, Y_{1:t}}\right]
    &=&\EE\left[\frac{b_{t|t-1}^2}{a_{t|t-1}^2(a_{t|t-1}-1)}\right]
    ~=~\EE\left[\frac{(p_ta_t+q_tb_t)^2}{((p_t+q_t)a_t)^2((p_t+q_t)a_t-1)}\right]\\
    &=&\frac{\EE\left[\left((1-q_t)a_t+q_tb_t\right)^2\right]}{a_t^2(a_t-1)}
    ~\le~\frac{\EE\left[b_t^2\right]}{a_t^2(a_t-1)}
    ~=~
        \EE\left[\Var{\Theta^{-1}_{t} \,|\, Y_{1:t}}\right].
  \end{eqnarray*}
The inequality follows from mean stationarity
\begin{equation*}
a_t= a_t \EE\left[\Theta_t^{-1}\right]
= a_t \EE\left[\EE\left[\left.\Theta_t^{-1}\right| Y_{1:t}\right]\right]
= a_t \EE\left[\frac{b_t}{a_t}\right] = \EE[b_t],
\end{equation*}
which together with Jensen's inequality implies
\begin{eqnarray*}
\EE\left[\left((1-q_t)a_t+q_tb_t\right)^2\right]
&=&(1-q_t)^2a_t^2 + 2q_t(1-q_t)a^2_t+
\EE\left[q^2_tb_t^2\right]
\\
&=&(1-q_t)^2\EE[b_t]^2 + 2q_t(1-q_t)\EE[b_t]^2+
\EE\left[q^2_tb_t^2\right]
\\
&\le&(1-q_t)^2\EE\left[b_t^2\right] + 2q_t(1-q_t)\EE\left[b_t^2\right]^2+
\EE\left[q^2_tb_t^2\right]
\\
&=&
\EE\left[b_t^2\right],
\end{eqnarray*}
and we have a strict inequality whenever $q_t<1$, which is equivalent to
$p_t>0$ under $p_t+q_t=1$.
This proves the claim.
\end{proof}

In the next result we prove that if, under the assumptions of
Theorem \ref{theorem decreasing}, the sequence $(q_t)_{t\ge 1}
\subseteq (0,1)$ is bounded away from zero and one, the inverse state-space
process has an asymptotically vanishing variance, i.e., asymptotically the randomness from the state-spaces is zero.

\begin{corollary}\label{corollary limit variance zero}
Consider Model \ref{model.2} with initialization $a_{1|0}>1$ and
with constraint \eqref{eq.4}. Moreover, assume $\delta <q_t^2
<1-\delta$ for some $\delta>0$ and all $t \ge 1$, and
$\sum_{t\ge 1}v_t=\infty$.
Then, we have
  \[
\lim_{t \to \infty}  \Var{\Theta^{-1}_{t}}=0.
  \]
\end{corollary}

\begin{proof}[Proof of Corollary \ref{corollary limit variance zero}.]
From Lemma \ref{cor.2} and under choice \eqref{eq.4} we have
\begin{equation*}
\Var{\Theta^{-1}_{t+1}}
= q_t^2\,
\Var{\Theta^{-1}_{t}}
+\left(1-q_t^2\right)\frac{1}{a_t-1}
\le  (1-\delta)\,
\Var{\Theta^{-1}_{t}}
+\left(1-\delta\right)\frac{1}{a_t-1},
\end{equation*}
for $t\ge 1$. Choose $t+s>t \ge 1$, then we have
using the monotonicity of Theorem \ref{theorem decreasing}
\begin{eqnarray*}
\Var{\Theta^{-1}_{t+s}}
&\le&  (1-\delta)^{s}\,
\Var{\Theta^{-1}_{t}}
+ \sum_{k=1}^s \left(1-\delta\right)^k\frac{1}{a_{t+s-k}-1}
\\
&\le&  (1-\delta)^{s}\,
\Var{\Theta^{-1}_{1}}
+ \left(\frac{1-\delta}{\delta}\right) \max_{1\le k\le s } \frac{1}{a_{t+s-k}-1}.
\end{eqnarray*}
Using \eqref{eq41} and \eqref{eq.update22} we have for $t\ge 2$
\begin{equation*}
a_t=a_{t|t-1}+v_t/\psi=(p_ {t-1}+q_{t-1})a_{t-1}+v_t/\psi=a_{t-1}+v_t/\psi
=a_{1|0} + \psi^{-1}\sum_{s=1}^t v_s.
\end{equation*}
Under our assumptions, this implies monotonicity of $(a_t)_{t \ge 1}$
with $\lim_{t \to \infty} a_t=\infty$. Plugging this into the previous inequality
we obtain upper bound
\begin{eqnarray*}
\Var{\Theta^{-1}_{t+s}}
&\le& (1-\delta)^{s}\,
\frac{1}{a_{1|0}-1}
+ \left(\frac{1-\delta}{\delta}\right) \frac{1}{a_{t}-1}.
\end{eqnarray*}
Using $\lim_{t \to \infty} a_t=\infty$, we can find for any $\varepsilon >0$
an index $t \ge 1$ such that the second term is
\begin{equation*}
\left(\frac{1-\delta}{\delta}\right) \frac{1}{a_{t}-1} < \varepsilon /2.
\end{equation*}
Moreover, there exists $s_0>0$ such that for all $s\ge s_0$
\begin{equation*}
(1-\delta)^{s}\,
\frac{1}{a_{1|0}-1} < \varepsilon /2.
\end{equation*}
But this implies that for all $u\ge s_0+t$ we have
\begin{equation*}
\Var{\Theta^{-1}_{u}}<\varepsilon.
\end{equation*}
Since $\varepsilon>0$ was arbitrary, the claim follows.
\end{proof}

\begin{remark}\normalfont
For $p_t \equiv 0$ and $q_t \equiv 1$, there is no state-space update from $t$ to $t+1$ meaning that
\[
\left.\Theta_{t}\right\vert_{Y_{1:t}} \eqd \left.\Theta_{t+1}\right\vert_{Y_{1:t}}.
\]
In this case, Model \ref{model.2} becomes a static-random effects model, which complies with  the \cite{BuhlmannStraub} assumptions, in particular, it has one static latent factor $\Theta_t \equiv \Theta \sim \Gamma(a_{1|0}, a_{1|0})$  \cite[page 44]{buhlmann2006course}.
\end{remark}

From the previous results, we conclude that Model \ref{model.2} allows for a wide range of variance behaviors for the inverse state-space dynamics $(\Theta_t^{-1})_{t \ge 1}$. The specific variance behavior should be determined by the data through model fitting and selection. This is discussed next.

%\section{Forecasting and evolutionary credibility}
\section{Model Fitting, forecasting, and evolutionary credibility}
\label{sec: Forecasting and evolutionary credibility}
\subsection{Evolutionary credibility and credibility formulas}
For experience rating, one is mainly interested in forecasting
the claim in the next period $Y_{t+1}$, given past information $Y_{1:t}$.
Under Model \ref{model.2}, we can compute the conditional expectation of the future claim,
given the past observations, by applying the tower property for conditional expectations. That is,
\begin{equation*}
\EE\left[\left. Y_{t+1} \right| Y_{1:t} \right]
=
\EE\left[\left.\EE\left[\left. Y_{t+1} \right| Y_{1:t}, \Theta_{1:t+1} \right]\right| Y_{1:t} \right]
= v_{t+1}\mu_{t+1}\,
\EE\left[\left.\Theta_{t+1}^{-1}\right| Y_{1:t} \right].
\end{equation*}
Using \eqref{thinning a} and mean stationarity we arrive at
\begin{equation}\label{credibility}
\EE\left[\left. Y_{t+1} \right| Y_{1:t} \right]
=v_{t+1}\mu_{t+1} \Big( \Delta_t \E{\Theta_{t}^{-1}\,|\, Y_{1:t}} +\left( 1-\Delta_t\right) \EE[\Theta_{t+1}^{-1}]\Big).
\end{equation}
Thus, the posterior prediction is a credibility weighted average between
the observation-driven estimate
$\E{\Theta_{t}^{-1}\,|\, Y_{1:t}}$ and the prior estimate
$\EE[\Theta_{t+1}^{-1}]=\EE[\Theta_{1}^{-1}]=1$, with credibility
weight, see \eqref{credibility Delta},
\begin{equation*}
\Delta_t = \frac{q_t}{q_t + p_t} ~\in~(0,1].
\end{equation*}
This recursive formula \eqref{credibility}, along with Bayesian filtering in \eqref{eq41}, provides a basis for deriving the evolutionary credibility
structure \cite[Chapter 9]{buhlmann2006course}. Specifically, starting from \eqref{eq.3400}, we obtain posterior mean
\begin{eqnarray*}
\E{\Theta_{t}^{-1}\,|\, Y_{1:t}} ~=~ \frac{b_t}{a_t}
&=& \frac{b_{t|t-1}+Y_t/(\mu_t\psi)}{a_{t|t-1}+v_t/\psi}
\\
&=& \begin{cases}\frac{v_t/\psi}{a_{t|t-1}+v_t/\psi}
\frac{1}{\mu_t}\frac{Y_t}{v_t}+
\frac{a_{t|t-1}}{a_{t|t-1}+v_t/\psi}
\frac{b_{t|t-1}}{a_{t|t-1}}, & \text{ for $v_t>0$;}\\
\frac{b_{t|t-1}}{a_{t|t-1}}, & \text{ for $v_t=0$.}\\
\end{cases}
\end{eqnarray*}
To simplify notation, for $Y_t = 0$ and $v_t = 0$, we define
\[
\frac{Y_t}{v_t} = 0,
\]
almost surely, which implies
\[
\frac{\E{\left. Y_t \right| Y_{1:t-1}}}{v_t} = 0.
\]
With this definition, the recursion can be compactly expressed as
\[
\E{\Theta_{t}^{-1}\,|\, Y_{1:t}} = \frac{v_t/\psi}{a_{t|t-1}+v_t/\psi}
\frac{1}{\mu_t}\frac{Y_t}{v_t} +
\frac{a_{t|t-1}}{a_{t|t-1}+v_t/\psi}\,
\E{\Theta_{t}^{-1}\,|\, Y_{1:t-1}}.
\]
Combining this with \eqref{stillgamma0} and \eqref{credibility}, we obtain the evolutionary credibility formula  for the normalized observations
$({Y_t}/({v_t}\mu_t))_{t \geq 1}$.

\begin{proposition}
Assume Model \ref{model.2} holds. Define the weights $z_t$ by
\[
z_t = \frac{v_t/\psi}{a_{t|t-1} + v_t/\psi}~ \in ~[0,1).
\]
For $t \ge 1$, the predictive conditional mean of the future claim is given by
\begin{eqnarray*}
\EE\left[\left. Y_{t+1} \right| Y_{1:t} \right]
&=& v_{t+1}\mu_{t+1}\,
\EE\left[\left.\Theta_{t+1}^{-1}\right| Y_{1:t} \right] \\
&=& v_{t+1}\mu_{t+1} \Big( \omega_{1,t}
\frac{1}{\mu_{t}}\frac{Y_t}{v_t} +
\omega_{2,t}
\frac{1}{\mu_t}\frac{\EE\left[\left.Y_t\right| Y_{1:t-1} \right]}{v_t} +
\omega_{3,t} \Big),
\end{eqnarray*}
where the credibility weights $\omega_{1,t}, \omega_{2,t}, \omega_{3,t} \geq 0$ satisfy $\omega_{1,t} + \omega_{2,t} + \omega_{3,t} = 1$
and are given by
\[
\omega_{1,t} = \Delta_t z_t, \quad \omega_{2,t} = \Delta_t (1 - z_t), \quad \text{and} \quad
\omega_{3,t} = 1 - \Delta_t.
\]
\end{proposition}

By recursively expanding the previous result we arrive at the following
corollary.

\begin{corollary}\label{explicit credibility}
Under Model \ref{model.2}, the predictive conditional mean can be represented by
\begin{equation*}
\begin{aligned}
\EE\left[\left. Y_{t+1} \right| Y_{1:t} \right]
&=     v_{t+1}\mu_{t+1} \left( \sum_{s=1}^{t} \left[\prod_{k=s+1}^{t} \omega_{2,k}\right]\left(\omega_{1,s}\,\left[\frac{1}{\mu_s}\frac{Y_s}{v_s}\right]
+ \omega_{3,s} \cdot 1\right) +\prod_{s=1}^t \omega_{2,s}\right),\\
\end{aligned}
\end{equation*}
for $t\ge 0$, where an empty product is set equal to 1 and an empty sum equal to zero.
\end{corollary}
By setting
$$\Delta_0=0\quad\hbox{and}\quad\frac{1}{\mu_0}\frac{Y_0}{v_0}=0,$$
the predictive conditional mean in Corollary \ref{explicit credibility} can be simplified as
\begin{equation*}
\begin{aligned}
\EE\left[\left. Y_{t+1} \right| Y_{1:t} \right]
&=
v_{t+1}\mu_{t+1} \left( \sum_{s=0}^{t} \left[\prod_{k=s+1}^{t} \omega_{2,k}\right]\left(\omega_{1,s}\,\left[\frac{1}{\mu_s}\frac{Y_s}{v_s}\right]
+ \omega_{3,s} \cdot 1\right)\right).\\
\end{aligned}
\end{equation*}
Thus, we obtain a credibility weighted average between
the observations $Y_{1:t}$ and the prior mean $\EE[\Theta_{t+1}^{-1}]=1$, and depending on the specific choices of the exogenous parameters this may lead to an exponentially decaying seniority weighting
in past claims; for seniority weighting see
\cite{pinquet2001allowance}.

\subsection{Likelihood function}

The joint density of the observations $Y_{1:t}$ under
Model \ref{model.2}
is given by the recursive formula
\begin{equation*}
    f(y_{1:t}) = f(y_t \mid y_{1:t-1}) f(y_{1:t-1}) = \prod_{s=1}^t f(y_s \mid y_{1:s-1}),
\end{equation*}
where $f(y_s \mid y_{1:s-1})$ for $2 \leq s \leq t$ represent the conditional density of $Y_s$ in $y_s$, for given $Y_{1:s-1}=y_{1:s-1}$, and $f(y_1) = f(y_1 \mid y_{1:0})$ represents the density of $Y_1$ in $y_1$. We have the following result.

\begin{lemma}\label{conditional marginals}
Under Model \ref{model.2}, for $s\ge 1$, we have the following holds.
\begin{itemize}
  \item[i.] The conditional density of $Y_s$ in $y_s>0$, given observations $Y_{1:s-1}=y_{1:s-1}$, is given by
\begin{equation*}
f(y_s\,|\,y_{1:s-1})
=\frac{\Gamma({v_s}/{\psi}+a_{s|s-1}+1)}{
\Gamma({v_s}/{\psi})\Gamma(a_{s|s-1}+1)}
\left(\frac{\frac{y_s}{\mu_s \psi}}
{\frac{y_s}{\mu_s\psi}+b_{s|s-1}}
\right)^{{v_s}/{\psi}}
\left(\frac{b_{s|s-1}}{\frac{y_s}{\mu_s\psi}+b_{s|s-1}}\right)^{a_{s|s-1}+1}
y_s^{-1}.
\end{equation*}
 \item[ii.] Define the random variable
\[
X_s =\frac{Y_s}{\mu_s \psi b_s}.
\]
The conditional density of $X_s$ in $x_s>0$, given observations $Y_{1:s-1}=y_{1:s-1}$, is given by
\[
f(x_s\,|\,y_{1:s-1})
=\frac{\Gamma\left({v_s}/{\psi}+ a_s+1 \right)}{\Gamma\left( {v_s}/{\psi} \right) \Gamma\left( a_s+1 \right)}
x_s^{{v_s}/{\psi}-1}(1+x_s)^{{v_s}/{\psi} +a_s+1}.
\]
This is a Beta-prime distribution with parameters ${v_s}/{\psi}>0$ and $a_s+1>0$, also called as Pearson type VI distribution \citep{johnson1995continuous}.
\end{itemize}
\end{lemma}

\begin{proof}[Proof of Lemma \ref{conditional marginals}.]
Using \eqref{equation2} and \eqref{stillgamma0}, we have for $s\ge 1$
and $y_s>0$
\begin{eqnarray*}
f(y_s\,|\,y_{1:s-1})&=& \int_{\theta_s} f(y_s\,|\,\theta_s)
f(\theta_s\,|\,y_{1:s-1})\, d\theta_s
\\
&=& \frac{\left(\frac{1}{\mu_s \psi}\right)^{\frac{v_s}{\psi}}}{\Gamma\left(\frac{v_s}{\psi}\right)}y_s^{\frac{v_s}{\psi}-1}
\frac{\left(b_{s|s-1}\right)^{a_{s|s-1}+1}}{\Gamma(a_{s|s-1}+1)}
\int_{\theta_s} \theta_s^{\frac{v_s}{\psi}+a_{s|s-1}+1-1}\, \exp\left\{ -\theta_s \left( \frac{y_s}{\mu_s\psi}+ b_{s|s-1}\right)\right\} d\theta_s
\\
&=&
\frac{\left(\frac{1}{\mu_s \psi}\right)^{\frac{v_s}{\psi}}}{\Gamma\left(\frac{v_s}{\psi}\right)}y_s^{\frac{v_s}{\psi}-1}
\frac{\left(b_{s|s-1}\right)^{a_{s|s-1}+1}}{\Gamma(a_{s|s-1}+1)}
\frac{\Gamma(\frac{v_s}{\psi}+a_{s|s-1}+1)}{
 \left(\frac{y_s}{\mu_s\psi}+b_{s|s-1}\right)^{
\frac{v_s}{\psi}+a_{s|s-1}+1}}
\\
&=&\frac{\Gamma(\frac{v_s}{\psi}+a_{s|s-1}+1)}{
\Gamma(\frac{v_s}{\psi})\Gamma(a_{s|s-1}+1)}
\left(\frac{\frac{y_s}{\mu_s \psi}}
{\frac{y_s}{\mu_s\psi}+b_{s|s-1}}
\right)^{\frac{v_s}{\psi}}
\left(\frac{b_{s|s-1}}{\frac{y_s}{\mu_s\psi}+b_{s|s-1}}\right)^{a_{s|s-1}+1}
y_s^{-1}.
\end{eqnarray*}
This completes the proof of part i.

The proof of the second part by a change of variables
for $X_s ={Y_s}/({\mu_s \psi b_s})$.
\end{proof}

Note that
\begin{equation*}
\lim_{y_s \to \infty} \frac{f(y_s\,|\,y_{1:s-1})}{y_s^{-(a_{s|s-1}+2)}}
%= c
%\end{equation*}
%where $c>0$ is defined as
%\begin{equation*}
%c=
~=~
\frac{\Gamma({v_s}/{\psi}+a_{s|s-1}+1)}{
\Gamma({v_s}/{\psi})\Gamma(a_{s|s-1}+1)}
\left(
\mu_s\psi b_{s|s-1}
\right)^{a_{s|s-1}+1} ~\in ~(0,\infty).
\end{equation*}
This implies that the conditional distribution of \(Y_s\), given \(Y_{1:s-1}\), is regularly varying at infinity with tail index \(-(a_{s|s-1}+2)\). Applying Karamata's theorem \citep{embrechts2013modelling} to the corresponding survival function shows that the conditional survival function of \(Y_s\), given \(Y_{1:s-1}\), is regularly varying at infinity with tail index \(-(a_{s|s-1}+1)\).
These asymptotic properties provide conditions for the finiteness of the conditional moments. In particular, if \(a_{s|s-1} > 0\), then \(\E{Y_s \mid Y_{1:s-1}}\) is finite, and if \(a_{s|s-1} > 1\), then \(\Var{Y_s \mid Y_{1:s-1}}\) is finite, which confirms the earlier conclusions regarding the existence of the conditional variance of
the inverse state-space \(\Theta_s^{-1}\).

In case of a unit exposure $v_t\equiv 1$ and a unit dispersion $\psi=1$,
the Gamma distribution \eqref{equation2} turns into the simpler exponential
distribution. As a consequence we receive the conditional distribution
\begin{equation*}
f(y_s\,|\,y_{1:s-1})=
\frac{a_{s|s-1}+1}{\mu_s b_{s|s-1}}
 \left(1+\frac{y_s}{\mu_sb_{s|s-1}}\right)^{-(a_{s|s-1}+2)}.
\end{equation*}
This is a Pareto Type II distribution, also called \cite{lomax1954business}
distribution, with a shape parameter $a_{s|s-1}+1>1$ and scale
parameter $\mu_sb_{s|s-1}>0$.

\section{Simulation study}
\label{section simulation study}
For our simulation study, we use the following data generation schemes. We assume there are $M$ distinct instances (insurance policies), that are observed for $T+1$ years. The records from first $T$ years are used for the model fitting, and the records of year $T+1$ are used for an out-of-sample validation. The exposure at time $1\le t \le T+1$ of instance $1\le i\le M$, $v_{i,t}$, follows the following distribution
$$
v_{i,t} = N_{i,t} + B_{i,t} ~\in ~\N_0,
$$
where $N_{i,t} \sim {\rm Poisson}(0.2\cdot (t+1))$, $B_{i,t} \sim {\rm Bernoulli}(1.2-0.2t)$, and $N_{i,t}$ and $B_{i,t}$ being independent
for all $t$ and $i$. The claim severity at time $t$ for instance $i$, $\mu_{i,t}$, is assumed to be known and selected by $\mu_{i,t} \sim {\rm Uniform}(2000, 4000)$. This gives us the exogenous sequences
$(v_{i,t}, \mu_{i,t})_{t=1}^{T+1}$ for all instances $1\le i \le M$.
Next, $(Y_{i,t}, \theta_{i,t})_{t=1}^{T+1}$ is generated under Model \ref{model.2}, and assuming \eqref{eq.112}, which implies a variance stationary process. In the construction, we further assume that $p_{i,t} = \kappa q_{i,t}$ for a constant $\kappa \ge 0$, which is equivalent to $\Delta_{i,t} \equiv \Delta$. Note that by construction, $Y_{i,t}=0$ if and only if $v_{i,t}=0$.

We generated 100 samples under the above data generation scheme with $M=5000$ and $T=5$, to estimate $a_{1|0}, \psi$, and $\Delta$. and subsequently predict $Y_{6}$ from $Y_{1:5}$, $v_{1:6}$, and $\mu_{1:6}$. For comparison purpose, we use the following six models:

\begin{itemize}
  \item Homogeneous independent: We assume that $\left. Y_{i,t} \right|_{Y_{i,1:t-1}}$ follows a Gamma distribution as in \eqref{equation2}, with $\mu_{i,t}=\mu$ and $\theta_{i,t}=1$ for $1 \le t\le  6$ and $1 \le i\le  5000$, where $\mu$ is estimated as the empirical weighted average of $(Y_{i,t})_{i,t}$ so that
    $$
    \widehat{\mu} = \frac{\sum_{i=1}^{M} \sum_{t=1}^T Y_{i,t}}{\sum_{i=1}^{M} \sum_{t=1}^T v_{i,t}}.
    $$
    \item Homogeneous B{\"u}hlmann model: We assume that $\left. Y_{i,t} \right|_{Y_{i,1:t-1}}$ follows a Gamma distribution as in \eqref{equation2}, $\mu_{i,t}=\mu$ and $\theta_{i,t}=\theta_{i}$ for $1 \le t\le  6$ and $1 \le i\le  5000$. Note that it is a special case of the next model (homogeneous SSM) with $\Delta=1$.
    \item Homogeneous state-space model (SSM): We assume that $\left. Y_{i,t} \right|_{Y_{i,1:t-1}}$ follows a Gamma distribution as in \eqref{equation2}, $\mu_{i,t}=\mu$ but allow for an evolution of $\theta_{i,1:6}$ under Model \ref{model.2}, assuming \eqref{eq.112} holds.
     \item Heterogeneous independent model: We assume that $\left. Y_{i,t} \right|_{Y_{i,1:t-1}}$ follows a Gamma distribution as in \eqref{equation2}, $\theta_{i,t}=1$ for $1 \le t\le  6$ and $1 \le i\le  5000$ but use given $\mu_{i,t}$'s.
    \item Heterogeneous B{\"u}hlmann model: We assume that $\left. Y_{i,t} \right|_{Y_{i,1:t-1}}$ follows a Gamma distribution as in \eqref{equation2}, $\theta_{i,t}=\theta_{i}$ for $1 \le t\le  6$ and $1 \le i\le  5000$ but we use the given $\mu_{i,t}$'s. Note that it is a special case of the next model (heterogeneous SSM) with $\Delta=1$.

    \item Heterogeneous SSM: We assume that $\left. Y_{i,t} \right|_{Y_{i,1:t-1}}$ follows a Gamma distribution as in \eqref{equation2} while we allow for an evolution of $\theta_{i,1:6}$ under Model \ref{model.2}, assuming \eqref{eq.112} holds, and we use teh given $\mu_{i,t}$ for $1 \le t\le  6$ and $1 \le i\le  5000$.
\end{itemize}

Table \ref{tab:simest} summarizes the estimation results of $a_{1|0}$, $\psi$ and $\Delta$ over the 100 simulations, both the mean estimates and the standard errors (in the parenthesis) are presented. Note that the homogeneous/heterogeneous independent models are excluded from the comparison as these model do not attempt to estimate $a_{1|0}$, $\psi$ and $\Delta$ by assuming $\theta$ is constantly equal to one. It can be seen that the heterogeneous SSM accurately estimates the values of $a_{1|0}$, $\psi$ and $\Delta$. In the case of the homogeneous SSM, these estimates are a little bit biased as it  has a misspecified marginal mean structure.

\begin{table}[!h]
\caption{Summary of estimation procedure for true values $a_{1|0}=3$ and $\psi=1$.\label{tab:simest}}
\centering
\begin{tabular}[t]{ccccccc}
\toprule
\multicolumn{1}{c}{ } & \multicolumn{3}{c}{$\Delta=0.5$} & \multicolumn{3}{c}{$\Delta=1.0$} \\
\cmidrule(l{3pt}r{3pt}){2-4} \cmidrule(l{3pt}r{3pt}){5-7}
  & $\widehat{a}_{1|0}$ & $\widehat{\psi}$ & $\widehat{\Delta}$ & $\widehat{a}_{1|0}$ & $\widehat{\psi}$ & $\widehat{\Delta}$\\
\midrule
Homogeneous & 5.5828 & 1.2222 & 1.0000 & 2.9460 & 1.0415 & 1.0000\\
B{\"u}hlmann & (0.2704) & (0.0114) & - & (0.1246) & (0.0088) & -\\\hline
Homogeneous & 2.7620 & 1.0151 & 0.4659 & 2.8364 & 1.0256 & 0.9750\\
SSM & (0.1273) & (0.0138) & (0.0238) & (0.1214) & (0.0104) & (0.0106)\\\hline
Heterogeneous & 5.8727 & 1.1858 & 1.0000 & 3.0151 & 1.0004 & 1.0000\\
B{\"u}hlmann & (0.2433) & (0.0108) & - & (0.1079) & (0.0083) & -\\\hline
Heterogeneous & 3.0279 & 1.0017 & 0.5027 & 2.9952 & 0.9977 & 0.9957\\
SSM & (0.1228) & (0.0135) & (0.0234) & (0.1077) & (0.0093) & (0.0065)\\
\bottomrule
\end{tabular}
\end{table}

Tables \ref{tab:simval0.5} and \ref{tab:simval1.0} display the out-of-sample validation results. The means and standard errors (in parenthesis) are received from the 100 iterations under the different values of $\Delta$, which imply different state-space structures for the evolution of the latent factors $(\Theta_{i,t})_{t=1}^{T+1}$.
The reported figures are the root-mean squared error (RMSE) and the Gamma deviance (GDEV) defined by
$$
\begin{aligned}
\text{RMSE}(\widehat{\mu}_{1:M}, Y_{1:M}|v_{1:M}) &= \sqrt{\frac{1}{M}\sum_{i=1}^M (\widehat{\mu}_{i}v_{i} - Y_{i})^2}, \\
\text{GDEV}(\widehat{\mu}_{1:M}, Y_{1:M}|v_{1:M}) &= 2\sum_{i=1}^M \left (-v_{i}  \log\left(\frac{Y_{i}}{\widehat{\mu}_{i} v_{i}} \right) + \frac{Y_{i}- \widehat{\mu}_{i}v_{i}}{\widehat{\mu}_{i}}\right),
\end{aligned}
$$
where we set $-v_{i} \log\left({Y_{i}}/(\widehat{\mu}_{i} v_{i}) \right)=0$ if $Y_i = v_i =0$, and we dropped the lower time
index $T+1$.

Table \ref{tab:simval0.5} considers the case $\Delta=0.5$, and it
shows that the true model (heterogeneous SSM) outperforms all other models in both measures. This is expected as it is the true model to describe the data generation scheme. In the case of $\Delta=1.0$,
Table \ref{tab:simval1.0},
the heterogeneous B{\"u}hlmann model is the true model so it shows the best predictive performance. Meanwhile, it is interesting to observe that the SSMs are capable of estimating the correct value of $\Delta$ and their predictive performance is almost comparable to the true B{\"u}hlmann model (in this case). Therefore, we conclude that the proposed SSMs are flexible enough to find the true model.

\begin{table}[!h]
\caption{Summary of out-of-sample validation with simulations ($a_{1|0}=3, \psi=1$, $\Delta=0.5$).\label{tab:simval0.5}}
\centering
\begin{tabular}[t]{ccccccc}
\toprule
\multicolumn{1}{c}{ } & \multicolumn{3}{c}{RMSE} & \multicolumn{3}{c}{GDEV} \\
\cmidrule(l{3pt}r{3pt}){2-4} \cmidrule(l{3pt}r{3pt}){5-7}
  & Independent & B{\"u}hlmann & SSM & Independent & B{\"u}hlmann & SSM\\
\midrule
Homogeneous & 6067.64 & 6133.54 & 5981.14 & 6680.86 & 6764.54 & 6499.17\\
 & (724.26) & (690.46) & (701.65) & (180.72) & (189.24) & (175.27)\\
 \hline
Heterogeneous & 5974.26 & 6025.61 & 5878.33 & 6417.36 & 6473.82 & 6229.60\\
 & (727.88) & (695.29) & (704.12) & (176.98) & (190.43) & (175.06)\\
\bottomrule
\end{tabular}
\end{table}

\begin{table}[!h]
\caption{Summary of out-of-sample validation with simulations ($a_{1|0}=3, \psi=1$, $\Delta=1.0$).\label{tab:simval1.0}}
\centering
\begin{tabular}[t]{ccccccc}
\toprule
\multicolumn{1}{c}{ } & \multicolumn{3}{c}{RMSE} & \multicolumn{3}{c}{GDEV} \\
\cmidrule(l{3pt}r{3pt}){2-4} \cmidrule(l{3pt}r{3pt}){5-7}
  & Independent & B{\"u}hlmann & SSM & Independent & B{\"u}hlmann & SSM\\
\midrule
Homogeneous & 5962.36 & 5116.05 & 5118.61 & 6639.29 & 5206.91 & 5212.95\\
 & (469.85) & (355.58) & (359.82) & (192.44) & (119.02) & (119.25)\\
 \hline
Heterogeneous & 5868.95 & 4939.25 & 4939.39 & 6371.98 & 4901.10 & 4901.98\\
 & (471.67) & (351.47) & (351.82) & (186.97) & (116.44) & (116.38)\\
\bottomrule
\end{tabular}
\end{table}

\section{Real data analysis}
\label{section real study}
To assess the applicability of the proposed framework, we analyze a
U.S.~based longitudinal outpatient visit dataset for the years 2019--2022, which is a part of the database called \href{https://meps.ahrq.gov/mepsweb/index.jsp}{Medical Expenditure Panel Survey (MEPS)}.
The original dataset, which contains $30,079$ records for $4,928$ patients over the years 2019--2022, is a combination of the following four tables in the MEPS database; \href{https://meps.ahrq.gov/mepsweb/data_stats/download_data_files_detail.jsp?cboPufNumber=HC-213F}{HC-213F}, \href{https://meps.ahrq.gov/mepsweb/data_stats/download_data_files_detail.jsp?cboPufNumber=HC-220F}{HC-220F}, \href{https://meps.ahrq.gov/mepsweb/data_stats/download_data_files_detail.jsp?cboPufNumber=HC-229F}{HC-229F}, and \href{https://meps.ahrq.gov/mepsweb/data_stats/download_data_files_detail.jsp?cboPufNumber=HC-239F}{HC-239F}. As displayed in Table \ref{tab:sample_original}, it contains various fields such as the identifier for a patient in the panel (\texttt{DUPERSID}), year and month information for the outpatient visit (\texttt{OPDATEYR} and \texttt{OPDATEMM}), categorical covariates that explain the characteristics of the outpatient visit (for example, \texttt{SEEDOC\_M18} and \texttt{LABTEST\_M18}; see the \href{https://meps.ahrq.gov/data_stats/download_data/pufs/h239f/h239fcb.pdf}{codebook} to check the comprehensive list of the covariates and their descriptions), and the total medical expenses charged for the visit (\texttt{OPTC[YR]X}).

\begin{table}[!h]
\caption{Sample rows from the original  dataset.\label{tab:sample_original}}
\centering
\begin{tabular}{ccccccc}
\toprule
\texttt{DUPERSID}   & \texttt{OPDATEYR} & \texttt{OPDATEMM} & \texttt{SEEDOC\_M18} & \texttt{LABTEST\_M18} & $\cdots$ & \texttt{OPTC[YR]X} \\ \midrule
2460002101 & 2021 & 8     & Y           & Y            & $\cdots$ & 6407.5        \\
2460002101 & 2022 & 10    & Y           & Y            & $\cdots$ & 5344          \\
$\vdots$        & $\vdots$  & $\vdots$   & $\vdots$         & $\vdots$          & $\ddots$ & $\vdots$           \\
2460004102 & 2021 & 3     & N           & N            & $\cdots$ & 83.01         \\
2460004102 & 2021 & 10    & N           & N            & $\cdots$ & 100.01  \\
\bottomrule
\end{tabular}
\end{table}

\subsection{Data pre-processing with a working model}
While the dataset contains rich information with repeated measurements, it requires some steps of pre-processing to be analyzed under the proposed framework. In the Generalized Smith--Miller Model, we implicitly assume that the state-space variables $(\Theta_t)_{t \ge 1}$ affect each period of the same length (for example, a year or a month) whereas the original dataset is not recorded in that fashion. In other words, the observed value of \texttt{OPTC[YR]X} in each record  in the original dataset corresponds to an individual payment $Z_{t[j]}$, as described in \eqref{eqn:distid}, not to $Y_t$ in \eqref{equation2} for a time period.

Therefore, the observed records for a person $i$ (identified by \texttt{DUPERSID}, here) need to be aggregated within pre-specified periods $t\ge 1$ to give us $(Y_{i,t}, v_{i,t}, \mu_{i,t})_{i,t}$, where $Y_{i,t}$ is the sum of \texttt{OPTC[YR]X} for person $i$ in period $t$,  $v_{i,t}$ is the total number of the outpatient visits for person $i$ in period $t$, and $\mu_{i,t}=\mathbb{E}[Z_{t[j]}]$ is the expectation of \texttt{OPTC[YR]X} per outpatient visit of person $i$ in period $t$. While the aggregations of $Y_{i,t}$ and $v_{i,t}$ are straightforward and do not involve any uncertainty, $\mu_{i,t}$ is unknown at this stage so we need a working model that estimates the marginal mean severity $\mu_{i,t}$ based on the observed values of the response variables $Z_{i,t[j]}$ (which is \texttt{OPTC[YR]X} in this case) and the corresponding covariates in the original dataset.

For this purpose, we exploit a preliminary Gamma
generalized linear model (GLM) assuming
$$
 Z_{i,t[j]}
   \sim \Gamma\left( \frac{1}{\psi}, {\frac{1}{\widetilde{\mu}_{i,t[j]}\psi}}\right), \qquad \text{ with regression mean: }  \log\left(\widetilde{\mu}_{i,t[j]}\right)= \mathbf{x}_{i,t[j]} \boldsymbol{\beta},
$$
where $\mathbf{x}_{i,t[j]}$ are pre-processed covariates from the original dataset as described in Table \ref{tab:desc}, and GLM
parameter $\boldsymbol{\beta}$. For the covariates pre-processing for the working model, we treated no response or missing values for \texttt{Care\_Category}, \texttt{Special\_Cond}, \texttt{Surgery}, \texttt{Prescription} as their proportions were negligible (less than 1\%). In the case of \texttt{Telehealth}, however, we treated the missing values or no response as a separate category, due to the reasons that for all responses from year 2019 records these are missing (as the telehealth indicator was not collected in the survey in year 2019), and about 25\% of records from years 2020--2022 indicate this value as missing.

\begin{table}[h!t!]
\caption{Description of the covariates in the working model. \label{tab:desc}}
\begin{center}
\resizebox{!}{0.2\linewidth}{
\begin{tabular}{l|lrrr}
\toprule
Variables & Description & Response & Proportions
 \\
\hline
\texttt{Doctor\_Type} & Type of medical professional for the visit           & GP & \multicolumn{1}{c}{7.00 \%} \\
 &      & Specialist & \multicolumn{1}{c}{39.27 \%} \\
 &  & Non-doctor & \multicolumn{1}{c}{53.73 \%} \\ \hline
\texttt{Care\_Category} & Main reason of the outpatient visit           & Diagnosis or Treatment & \multicolumn{1}{c}{51.67 \%} \\
&       & Others & \multicolumn{1}{c}{48.33 \%} \\ \hline
\texttt{Special\_Cond} & Indicator for an existing special condition for the visit           & Yes & \multicolumn{1}{c}{83.62 \%} \\
&       & No & \multicolumn{1}{c}{16.38 \%} \\ \hline
\texttt{Surgery} & Whether the visit involved a surgery        & Yes & \multicolumn{1}{c}{8.68 \%} \\
&       & No & \multicolumn{1}{c}{91.32 \%} \\ \hline
\texttt{Prescription} & Whether any medicine prescribed for the visit     & Yes & \multicolumn{1}{c}{8.33 \%} \\
&       & No & \multicolumn{1}{c}{91.67 \%} \\ \hline
\texttt{Telehealth} & Whether the visit was a telehealth event   & Yes & \multicolumn{1}{c}{7.61 \%} \\
&       & No & \multicolumn{1}{c}{40.35 \%} \\
&       & Unknown & \multicolumn{1}{c}{52.04 \%} \\
\hline \hline
\end{tabular}}
\end{center}
\end{table}

The estimation results of the working model are summarized in Table \ref{tab:marginal_est}, which also provides some intuitive interpretations. For example, there is expected to incur higher charge if an outpatient visit involved a meeting with a specialist, rather than a general or family doctor. It is also natural to expect that outpatient visits with surgeries could incur higher charges than those without surgeries. Lastly, it is also shown that a telehealth outpatient visit gives less charges than an in-person outpatient visit.

\begin{table}[!h]
\caption{Summary of the estimated regression coefficients for the working severity model. \label{tab:marginal_est}}
\centering
\begin{tabular}[t]{lcc}
\toprule
  & Estimate & p-value\\
\midrule
(Intercept) & 7.4351 & 0.0000\\
\texttt{Doctor\_Type}: Non-doctor & 0.1304 & 0.0338\\
\texttt{Doctor\_Type}: Specialist & 0.5515 & 0.0000\\
\texttt{Care\_Category}: Others & -0.1390 & 0.0000\\
\texttt{Special\_Cond}: Yes & 0.1912 & 0.0000\\
\texttt{Surgery}: Yes & 1.6888 & 0.0000\\
\texttt{Prescription}: Yes & -0.1080 & 0.0531\\
\texttt{Telehealth}: Unknown & -0.1501 & 0.0000\\
\texttt{Telehealth}: Yes & -0.8352 & 0.0000\\
\bottomrule
\end{tabular}
\end{table}

Note that it is possible that $\mathbf{x}_{i,t[j]}$ may not be identical for all $j=1, \ldots, v_{i,t}$ as one patient can have multiple outpatient visits in a time period due to different reasons. Therefore, we summarized the estimates from the working model $\widetilde{\mu}_{i,t[j]}$ to obtain a reasonable estimate of $\mu_{i,t}$ as $\widehat{\mu}_{i,t} = \frac{1}{v_{i,t}}\sum_{j=1}^{v_{i,t}}\widetilde{\mu}_{i,t[j]}
$. Sample rows from the resulting dataset in the form of $(Y_{i,t}, v_{i,t}, \mu_{i,t})_{i,t}$ are displayed in Table \ref{tab:sample_processed}. Note that we summarized the observations per half-year intervals so that the pre-processed dataset has $39,424 = 4928 \times 8$ records, from $4,928$ distinct patients and $8$ observational periods (from 2019H1 to 2022H2).

\begin{table}[!h]
\caption{Sample rows from the pre-processed  dataset.\label{tab:sample_processed}}
\centering
\begin{tabular}{ccccc}
\toprule
\texttt{ID}   & \texttt{HY} & $Y_{i,t}$ & $v_{i,t}$ & $\widehat{\mu}_{i,t}$  \\ \midrule
2460002101 & 2019H1 & 1209.00   &  1         & 1750.443                \\
2460002101 & 2019H2 & 13712.17   &  4         & 2011.496                \\
2460002101 & 2020H1 & 0   &  0         & -                \\
2460002101 & 2020H2 & 0   &  0         & -                \\
$\vdots$        & $\vdots$  & $\vdots$   & $\vdots$         & $\vdots$                  \\
\bottomrule
\end{tabular}
\end{table}

We conclude this subsection with two remarks. Firstly, it is not our primary concern to find either the best working model or the best feature engineering to estimate the marginal severity, \texttt{OPTC[YR]X}. Possibly, one can conduct a more sophisticated analysis to find the best set of engineered features and/or try different predictive models (including but not limited to random forests or neural networks). However, we believe that the current working model is reasonable enough to quantify the underlying impacts of the available covariate information on the response variables, with some intuitive explanations. Note also that our main goal is to evaluate the usefulness of the proposed SSM framework in comparison to the existing models that could describe the serial correlations among the repeated observations, given $\widehat{\mu}_{i,t}$ from some reasonably good working model.

\subsection{Analysis of the summarized dataset with the proposed model and some benchmarks}

For comparison purposes, we applied the six models (Homogeneous independent, Homogeneous B{\"u}hlmann, Homogeneous SSM,
Heterogeneous independent, Heterogeneous B{\"u}hlmann, and Heterogeneous SSM) that were described in Section \ref{section simulation study}. Recall that all of the aforementioned models implicitly or explicitly assume a stationary variance for $(\Theta^{-1}_{i,t})_{t\ge 1}$ by satisfying \eqref{eq.112}. We also assume $p_{i,t} = \kappa q_{i,t}$ for a constant $\kappa \ge 0$. Thus, the dependence structures are parameterized by the following three parameters $a_{1|0}$, $\psi$ and $\Delta$.

Table \ref{tab:actest} summarizes the estimation results of $a_{1|0}$, $\psi$ and $\Delta$ with the aggregated version of actual data as discribed above. Again, the homogeneous/heterogeneous independent models are excluded from the comparison as these models do not attempt to estimate $a_{1|0}$, $\psi$ and $\Delta$. by assuming $\theta$ is constantly equal to one. It is shown that the estimated values of $\Delta$ under the homogeneous/heterogeneous SSMs  are around 0.3, which implies
a rather fast decay of the impact of the latent factors $(\Theta_{i,t})_{t\ge 1}$.

As mentioned in the proof of Proposition \ref{lemma 1 mean stationarity}, there is an AR(1) type identity in the relationship between $\Theta^{-1}_{i,t}$ and $\Theta^{-1}_{i,t+1}$, that is,
$$
\EE \left[\left.\Theta^{-1}_{i,t+1}\right| Y_{i,1:t}\right] =  \Delta\,\EE \left[\left.\Theta^{-1}_{i,t}\right| Y_{i,1:t}\right]
+ (1-\Delta).
$$
That being said, if $\Delta \simeq 0.3$, then the approximate correlation between $\Theta^{-1}_{i,t}$ and $\Theta^{-1}_{i,t+3}$ is only about $\Delta^3 = 2.7\%$, which is indeed negligible compared to its full magnitude of 100\%. Intuitively, this means that many of the outpatient visits do not occur over a longer time period (possibly for the same reason).

\begin{table}[!h]
\caption{Summary of the parameter estimation on the real dataset. \label{tab:actest}}
\centering
\begin{tabular}[t]{lccc}
\toprule
  & $\widehat{a}_{1|0}$ & $\widehat{\psi}$ & $\widehat{\Delta}$\\
\midrule
Homogeneous B{\"u}hlmann & 0.8257 & 2.0394 & 1.0000\\
Homogeneous SSM & 0.4166 & 0.8746 & 0.3262\\
Heterogeneous B{\"u}hlmann & 1.4543 & 1.7334 & 1.0000\\
Heterogeneous SSM & 0.7150 & 0.8257 & 0.2825\\
\bottomrule
\end{tabular}
\end{table}

Such an impact from seemingly fast-decaying state-space variables are pronounced in the out-of-sample validation results in Table \ref{tab:actval}. It is shown that whether we assume a homogeneous or heterogeneous external mean model, the predictive performance of the B{\"u}hlmann models (which is analogous to assume $\Delta=1$) are worse than that of the independent models (which is analogous to assume $\Delta=0$). On the other hand, the proposed SSMs by far show the best predictive performance by capturing the correct magnitude of the decaying factor of the inverse state-space variables $\Theta^{-1}_{i,t}$ from the data.

\begin{table}[!h]
\caption{Summary of out-of-sample validation on the real dataset.\label{tab:actval}}
\centering
\begin{tabular}[t]{ccccccc}
\toprule
\multicolumn{1}{c}{ } & \multicolumn{3}{c}{RMSE} & \multicolumn{3}{c}{GDEV} \\
\cmidrule(l{3pt}r{3pt}){2-4} \cmidrule(l{3pt}r{3pt}){5-7}
  & Independent & B{\"u}hlmann & SSM & Independent & B{\"u}hlmann & SSM\\
\midrule
Homogenous & 11562.25 & 12688.83 & 10715.44 & 4512.37 & 4612.77 & 4185.06\\
Heterogeneous & 10809.82 & 10940.15 & 10173.28 & 3324.28 & 3532.69 & 3118.87\\
\bottomrule
\end{tabular}
\end{table}

\section{Summary}
\label{Concluding remarks}

Observation-driven state-space models are widely used in time-series modeling due to their analytical tractability in many cases. For claim counts, \cite{harvey1989time} introduced a Poisson-Gamma observation-driven state-space model that is fully analytically tractable, and \cite{ahn2023classification} later extended this model to accommodate flexible variance behavior in the latent state-space dynamics.

For claim size modeling or positive continuous random response modeling, the Gamma distribution is a natural and often preferred choice. \cite{smith1986non} proposed the Gamma-Gamma observation-driven state-space model, offering a fully tractable modeling framework. However, a limitation of the \cite{smith1986non} model is that its variance behavior is constrained to be increasing. This paper aims to generalize the Gamma-Gamma observation-driven state-space model to allow for flexible variance behavior, addressing this limitation.

Looking ahead, it would be valuable to develop a more rigorous classification of parameter-driven versus observation-driven state-space models. In particular, understanding the conditions under which a given response process $(Y_{t})_{t\geq 1}$ admits both parameter-driven and observation-driven state-space representations would be of great interest. Naturally, this involves addressing the identifiability issues, which are not yet fully resolved.

The Poisson-Gamma and Gamma-Gamma cases are part of the exponential dispersion family (EDF) with conjugate priors. A natural next step would be to investigate whether our results extend to the entire class of EDF models with conjugate priors. Additionally, the Gamma state-space process possesses several desirable properties, as exploited in \cite{ahn2023classification}. It would be intriguing to determine the extent to which these properties are necessary to maintain analytical tractability.

\bibliographystyle{apalike}
\bibliography{CTE_Bib_HIX2}

\begin{thebibliography}{}

\bibitem[Ahn et~al., 2021]{ahn2021order}
Ahn, J.~Y., Jeong, H., and Lu, Y. (2021).
\newblock On the ordering of credibility factors.
\newblock {\em Insurance: Mathematics and Economics}, 101:626--638.

\bibitem[Ahn et~al., 2023a]{youn2023simple}
Ahn, J.~Y., Jeong, H., and Lu, Y. (2023a).
\newblock A simple {B}ayesian state-space approach to the collective risk models.
\newblock {\em Scandinavian Actuarial Journal}, 2023(5):509--529.

\bibitem[Ahn et~al., 2023b]{ahn2023classification}
Ahn, J.~Y., Jeong, H., Lu, Y., and W{\"u}thrich, M.~V. (2023b).
\newblock A classification of observation-driven state-space count models for panel data.
\newblock {\em arXiv:2308.16058}.

\bibitem[Anderson and Moore, 2005]{anderson2005optimal}
Anderson, B.~D. and Moore, J.~B. (2005).
\newblock {\em Optimal filtering}.
\newblock Courier Corporation.

\bibitem[Arulampalam et~al., 2002]{arulampalam2002tutorial}
Arulampalam, M.~S., Maskell, S., Gordon, N., and Clapp, T. (2002).
\newblock A tutorial on particle filters for online nonlinear/non-{G}aussian {B}ayesian tracking.
\newblock {\em IEEE Transactions on Signal Processing}, 50(2):174--188.

\bibitem[Bichsel, 1964]{bichsel1964erfahrungs}
Bichsel, F. (1964).
\newblock Erfahrungs-{T}arifierung in der {M}otorfahrzeughaftpflicht-{V}ersicherung.
\newblock {\em Bulletin of the Swiss Association of Actuaries}, 1964:119--130.

\bibitem[Bolanc{\'e} et~al., 2007]{bolance2007greatest}
Bolanc{\'e}, C., Denuit, M., Guill{\'e}n, M., and Lambert, P. (2007).
\newblock Greatest accuracy credibility with dynamic heterogeneity: the {H}arvey-{F}ernandes model.
\newblock {\em Belgian Actuarial Bulletin}, 7(1):14--18.

\bibitem[Bolanc{\'e} et~al., 2003]{bolance2003time}
Bolanc{\'e}, C., Guill{\'e}n, M., and Pinquet, J. (2003).
\newblock Time-varying credibility for frequency risk models: estimation and tests for autoregressive specifications on the random effects.
\newblock {\em Insurance: Mathematics and Economics}, 33(2):273--282.

\bibitem[B{\"u}hlmann, 1967]{buhlmann1967experience}
B{\"u}hlmann, H. (1967).
\newblock Experience rating and credibility.
\newblock {\em ASTIN Bulletin - The Journal of the IAA}, 4(3):199--207.

\bibitem[B{\"u}hlmann and Gisler, 2006]{buhlmann2006course}
B{\"u}hlmann, H. and Gisler, A. (2006).
\newblock {\em A course in credibility theory and its applications}.
\newblock Springer Science \& Business Media.

\bibitem[B{\"u}hlmann and Straub, 1970]{BuhlmannStraub}
B{\"u}hlmann, H. and Straub, E. (1970).
\newblock Glaubw{\"u}rdigkeit f{\"u}r {S}chadens{\"a}tze.
\newblock {\em Bulletin of the Swiss Association of Actuaries}, 70(1):111--133.

\bibitem[Cox, 1981]{cox1981statistical}
Cox, D.~R. (1981).
\newblock Statistical analysis of time series: Some recent developments.
\newblock {\em Scandinavian Journal of Statistics}, 8(2):93--115.

\bibitem[Doucet et~al., 2000]{doucet2000sequential}
Doucet, A., Godsill, S., and Andrieu, C. (2000).
\newblock On sequential {M}onte {C}arlo sampling methods for {B}ayesian filtering.
\newblock {\em Statistics and Computing}, 10:197--208.

\bibitem[Embrechts et~al., 2013]{embrechts2013modelling}
Embrechts, P., Kl{\"u}ppelberg, C., and Mikosch, T. (2013).
\newblock {\em Modelling extremal events: for insurance and finance}, volume~33.
\newblock Springer Science \& Business Media.

\bibitem[Harvey and Fernandes, 1989]{harvey1989time}
Harvey, A.~C. and Fernandes, C. (1989).
\newblock Time series models for count or qualitative observations.
\newblock {\em Journal of Business \& Economic Statistics}, 7(4):407--417.

\bibitem[Johnson et~al., 1995]{johnson1995continuous}
Johnson, N.~L., Kotz, S., and Balakrishnan, N. (1995).
\newblock {\em Continuous univariate distributions, volume 2}, volume 289.
\newblock John Wiley \& Sons.

\bibitem[Jorgensen, 1997]{jorgensen1997theory}
Jorgensen, B. (1997).
\newblock {\em The theory of dispersion models}.
\newblock CRC Press.

\bibitem[Kalman, 1960]{Kalman1960}
Kalman, R.~E. (1960).
\newblock A new approach to linear filtering and prediction problems.
\newblock {\em Journal of Basic Engineering}, 82(1):35--45.

\bibitem[Laird and Ware, 1982]{laird1982random}
Laird, N.~M. and Ware, J.~H. (1982).
\newblock Random-effects models for longitudinal data.
\newblock {\em Biometrics}, pages 963--974.

\bibitem[Lee and Nelder, 1996]{lee1996hierarchical}
Lee, Y. and Nelder, J.~A. (1996).
\newblock Hierarchical generalized linear models.
\newblock {\em Journal of the Royal Statistical Society Series B: Statistical Methodology}, 58(4):619--656.

\bibitem[Lomax, 1954]{lomax1954business}
Lomax, K.~S. (1954).
\newblock Business failures: Another example of the analysis of failure data.
\newblock {\em Journal of the American Statistical Association}, 49(268):847--852.

\bibitem[Pinquet, 2020a]{pinquet2020poisson}
Pinquet, J. (2020a).
\newblock Poisson models with dynamic random effects and nonnegative credibilities per period.
\newblock {\em ASTIN Bulletin - The Journal of the IAA}, 50(2):585--618.

\bibitem[Pinquet, 2020b]{pinquet2020positivity}
Pinquet, J. (2020b).
\newblock Positivity properties of the arfima (0, d, 0) specifications and credibility analysis of frequency risks.
\newblock {\em Insurance: Mathematics and Economics}, 95:159--165.

\bibitem[Pinquet et~al., 2001]{pinquet2001allowance}
Pinquet, J., Guill{\'e}n, M., and Bolanc{\'e}, C. (2001).
\newblock Allowance for the age of claims in bonus-malus systems.
\newblock {\em ASTIN Bulletin - The Journal of the IAA}, 31(2):337--348.

\bibitem[Smith and Miller, 1986]{smith1986non}
Smith, R. and Miller, J. (1986).
\newblock {A non-Gaussian state space model and application to prediction of records}.
\newblock {\em Journal of the Royal Statistical Society: Series B}, 48(1):79--88.

\bibitem[Whitney, 1918]{whitney1918}
Whitney, A.~W. (1918).
\newblock The theory of experience rating.
\newblock {\em Proceedings of the Casualty Actuarial Society}, 4:274--292.

\end{thebibliography}

\end{document}